\documentclass[11pt,a4paper]{article}
\usepackage{amssymb}
\usepackage{amsmath}
\usepackage{amsfonts}
\usepackage{bbm}
\usepackage{amsthm}
\usepackage{mathrsfs}
\usepackage{hyperref}
\usepackage{color}
\usepackage[margin=2.41cm]{geometry}
\usepackage[all,cmtip]{xy}
\usepackage[utf8]{inputenc}
\usepackage{graphicx}
\usepackage{varwidth}
\usepackage{comment}

\usepackage{upgreek}
\usepackage{rotating}

\usepackage{tikz}
\usetikzlibrary{shapes.geometric}

\definecolor{darkred}{rgb}{0.8,0.1,0.1}
\hypersetup{
     colorlinks=false,         
     linkcolor=darkred,
     citecolor=blue,
}

\theoremstyle{plain}
\newtheorem{theo}{Theorem}[section]
\newtheorem{lem}[theo]{Lemma}
\newtheorem{propo}[theo]{Proposition}
\newtheorem{cor}[theo]{Corollary}

\theoremstyle{definition}
\newtheorem{defi}[theo]{Definition}

\newtheorem{exx}[theo]{Example}
\newenvironment{ex}
  {\pushQED{\qed}\exx}
  {\popQED\endexx}

\newtheorem{remm}[theo]{Remark}
\newenvironment{rem}
  {\pushQED{\qed}\remm}
  {\popQED\endremm}

\numberwithin{equation}{section}

\def\nn{\nonumber}

\def\bbR{\mathbb{R}}
\def\bbC{\mathbb{C}}

\def\bbZ{\mathbb{Z}}

\def\Hom{\mathrm{Hom}}

\def\id{\mathrm{id}}
\def\supp{\mathrm{supp}}

\def\vol{\mathrm{vol}}

\def\cc{\mathrm{c}}

\def\1{I}
\def\oone{\mathbbm{1}}

\def\Cauchy{\mathcal{C}}
\def\Reg{\mathcal{R}}
\def\Lan{\operatorname{Lan}}

\def\Alg{\mathbf{Alg}}
\def\Vec{\mathbf{Vec}}

\def\QFT{\mathbf{QFT}}
\def\qft{\mathbf{qft}}
\def\IQFT{\mathbf{IQFT}}

\def\CC{\mathbf{C}}
\def\DD{\mathbf{D}}

\def\ext{\operatorname{ext}}
\def\res{\operatorname{res}}

\def\AAA{\mathfrak{A}}

\def\BBB{\mathfrak{B}}

\def\III{\mathfrak{I}}
\def\KKK{\mathfrak{K}}

\newcommand\und[1]{\underline{#1}}

\DeclareMathOperator*{\interior}{\mathrm{int}}

\def\sk{\vspace{2mm}}

\makeatletter
\let\@fnsymbol\@alph
\makeatother

\title{%
Algebraic quantum field theory on \\
spacetimes with timelike boundary
}

\author{%
Marco Benini$^{1,a}$,\ 
Claudio Dappiaggi$^{2,b}$\ and\
Alexander Schenkel$^{3,c}$\vspace{4mm}\\
{\small ${}^1$ Fachbereich Mathematik, Universit\"at Hamburg,}\\
{\small Bundesstr.~55, 20146 Hamburg, Germany.}\vspace{3mm}\\
{\small ${}^2$ Dipartimento di Fisica, Universit\`a di Pavia 
\& INFN, Sezione di Pavia,}\\
{\small Via Bassi~6, 27100 Pavia, Italy.}\vspace{3mm}\\
{\small ${}^3$ School of Mathematical Sciences, University of Nottingham,}\\
{\small University Park, Nottingham NG7 2RD, United Kingdom.}\vspace{5mm}\\
{\small \begin{tabular}{ll}
Email: & ${}^a$~\texttt{marco.benini@uni-hamburg.de}\\
& ${}^b$~\texttt{claudio.dappiaggi@unipv.it}\\
& ${}^c$~\texttt{alexander.schenkel@nottingham.ac.uk}\vspace{3mm}
\end{tabular}
}
}

\date{May 2018}

\begin{document}

\maketitle

\begin{center}
{\em Dedicated to Klaus Fredenhagen on the occasion of his 70th birthday}\vspace{2mm}
\end{center}

\begin{abstract}
We analyze quantum field theories on spacetimes $M$ with timelike boundary from a model-independent perspective. We construct an adjunction which describes a universal extension to the whole spacetime $M$ of theories defined only on the interior $\mathrm{int}M$. The unit of this adjunction is a natural isomorphism, which implies that our universal extension satisfies Kay's F-locality property. Our main result is the following characterization theorem: Every quantum field theory on $M$ that is additive from the interior (i.e.\ generated by observables localized in the interior) admits a presentation by a quantum field theory on the interior $\mathrm{int}M$ and an ideal of its universal extension that is trivial on the interior. We shall illustrate our constructions by applying them to the free Klein-Gordon field.
\end{abstract}

\vspace{2mm}

\paragraph*{Report no.:} ZMP-HH/17-31, Hamburger Beitr\"age zur Mathematik Nr.\ 714



\newpage 

\tableofcontents


\section{\label{sec:intro}Introduction and summary}
Algebraic quantum field theory is a powerful and far developed 
framework to address model-independent aspects of quantum field 
theories on Minkowski spacetime \cite{Haag} 
and more generally on globally hyperbolic spacetimes \cite{Brunetti}.
In addition to establishing the axiomatic foundations for quantum field theory,
the algebraic approach has provided a variety of 
mathematically rigorous constructions of non-interacting
models, see e.g.\ the reviews \cite{BD,BDH,BGP}, and more interestingly 
also perturbatively interacting quantum field theories, see e.g.\ the recent 
monograph \cite{Rejzner}. It is worth emphasizing that many of the techniques
involved in such constructions, e.g.\ existence and uniqueness of Green's operators and
the singular structure of propagators, crucially rely on the hypothesis that
the spacetime is globally hyperbolic and has empty boundary.
\sk

Even though globally hyperbolic spacetimes have plenty of applications
to physics, there exist also important and interesting situations
which require non-globally hyperbolic spacetimes, possibly with a non-trivial boundary. 
On the one hand, recent developments in high energy physics and string theory are 
strongly focused on anti-de Sitter spacetime,
which is not globally hyperbolic and has a (conformal) timelike boundary. On the other hand,
experimental setups for studying the Casimir effect confine quantum field theories
between several metal plates (or other shapes), which may be modeled theoretically by introducing 
timelike boundaries to the system. This immediately prompts the question whether the rigorous 
framework of algebraic quantum field theory admits a generalization to cover such scenarios.
\sk

Most existing works on algebraic quantum field theory on spacetimes
with a timelike boundary focus on the construction of concrete examples,
such as the free Klein-Gordon field on simple classes of spacetimes. 
The basic strategy employed in such constructions is to analyze the initial
value problem on a given spacetime with timelike boundary, which has
to be supplemented by suitable boundary conditions. Different choices
of boundary conditions lead to different Green's operators for the equation
of motion, which is in sharp contrast to the well-known existence and 
uniqueness results on globally hyperbolic spacetimes with empty boundary.
Recent works addressing this problem are \cite{Zahn:2015due} and \cite{Ishibashi:2003jd,Ishibashi:2004wx}, 
the latter extending the analysis of \cite{Wald:1980jn}. 
For specific choices of boundary conditions, there exist successful constructions
of algebraic quantum field theories on spacetimes with timelike boundary, see e.g.\
\cite{Casimir,Dappiaggi2,Dappiaggi3,Dappiaggi1}.
The main message of these works is that the algebraic approach 
is versatile enough to account also for these models, although some key structures, such as 
for example the notion of Hadamard states \cite{Dappiaggi3,Wrochna:2016ruq}, 
should be modified accordingly.
\sk

Unfortunately, model-independent results on algebraic quantum field theory
on spacetimes with timelike boundary are more scarce. There are, however,
some notable and very interesting works in this direction: On the one hand, Rehren's proposal for 
algebraic holography \cite{Rehren} initiated the rigorous study of quantum field
theories on the anti-de Sitter spacetime. This has been further elaborated in
\cite{Duetsch:2002hc} and extended to asymptotically AdS spacetimes in \cite{Ribeiro}.
On the other hand, inspired by Fredenhagen's universal algebra \cite{Fre1,Fre2,Fre3}, 
a very interesting construction and analysis of \textit{global} algebras of observables 
on spacetimes with timelike boundaries has been performed in \cite{Sommer}. The most 
notable outcome is the existence  of a relationship between maximal ideals of this algebra and 
boundary conditions, a result which has been of inspiration for this work.
\sk

In the present paper we shall analyze quantum field theories on spacetimes 
with timelike boundary from a model-independent perspective. 
We are mainly interested in understanding and proving structural results for
whole categories of quantum field theories, in contrast to focusing on
particular theories. Such questions can be naturally addressed by using 
techniques from the recently developed operadic approach 
to algebraic quantum field theory \cite{operad}.
Let us describe rather informally the
basic idea of our construction and its implications: 
Given a spacetime $M$ with timelike
boundary, an algebraic quantum field theory on $M$ is a functor
$\BBB: \Reg_M\to\Alg$ assigning algebras of observables to suitable regions $U\subseteq M$
(possibly intersecting the boundary), which satisfies the causality and time-slice axioms.
We denote by $\QFT(M)$ the category of algebraic quantum field theories on $M$.
Denoting the full subcategory of regions in the interior of $M$ by $\Reg_{\interior{M}} \subseteq \Reg_M$,
we may restrict any theory $\BBB\in\QFT(M)$ to a theory $\res\BBB\in \QFT(\interior{M})$ defined 
only on the interior regions. Notice that it is in practice much easier to analyze and construct
theories on $\interior{M}$ as opposed to theories on the whole spacetime $M$. 
This is because the former are postulated to be insensitive to the boundary by 
Kay's F-locality principle \cite{Flocality}. As a first result we shall construct a left adjoint of the restriction functor
$\res : \QFT(M)\to \QFT(\interior{M})$, which we call the universal extension
functor $\ext : \QFT(\interior{M})\to\QFT(M)$. This means that given any theory
$\AAA \in\QFT(\interior{M})$ that is defined only on the interior regions in $M$,
we obtain a universal extension $\ext \AAA\in\QFT(M)$ to all regions in $M$,
including those that intersect the boundary. It is worth to emphasize that
the adjective \textit{universal} above refers to the categorical concept of
universal properties. Below we explain in which sense $\ext$ is also ``universal''
in a more physical meaning of the word.
\sk

It is crucial to emphasize that our universal
extension $\ext\AAA\in\QFT(M)$ is always a bona fide algebraic quantum field theory in the sense
that it satisfies the causality and time-slice axioms. 
This is granted by the operadic approach 
to algebraic quantum field theory of \cite{operad}.  
In particular, the $\ext \dashv \res$ adjunction investigated in the present paper 
is one concrete instance of a whole family 
of adjunctions between categories of algebraic quantum field theories 
that naturally arise within the theory of colored operads and algebras over them.
\sk

A far reaching implication of the above mentioned $\ext \dashv \res$ adjunction is 
a characterization theorem that we shall establish
for quantum field theories on spacetimes with timelike boundary. Given any theory $\BBB\in\QFT(M)$
on a spacetime $M$ with timelike boundary, we can restrict and universally extend
to obtain another such theory $\ext\res \BBB\in\QFT(M)$. The adjunction also provides
us with a natural comparison map between these theories, namely
the counit $\epsilon_{\BBB}: \ext\res \BBB\to \BBB$ of the adjunction. Our result in 
Theorem \ref{theo:quotientVSadditivity} and Corollary \ref{cor:quotientVSadditivity}
is that $\epsilon_\BBB$ induces an isomorphism $\ext\res \BBB /\ker\epsilon_\BBB \cong \BBB$
of quantum field theories if and only if $\BBB$ is \textit{additive 
from the interior} as formalized in Definition \ref{def:additive}. 
The latter condition axiomatises the heuristic idea that 
the theory $\BBB$ has no degrees of freedom that are localized
on the boundary of $M$, i.e.\ all its observables may be generated
by observables supported in the interior of $M$.
Notice that the results in Theorem \ref{theo:quotientVSadditivity} 
and Corollary \ref{cor:quotientVSadditivity} give the adjective \textit{universal} 
also a physical meaning in the sense that the extensions are sufficiently
large such that any additive theory can be recovered by a quotient.
We strengthen this result in Theorem \ref{theo:IQFTisQFT} by constructing
an equivalence between the category of additive quantum field theories on $M$ 
and a category of pairs $(\AAA,\III)$ consisting of a theory $\AAA\in\QFT(\interior{M})$ on the interior
and an ideal $\III\subseteq \ext\AAA$ of the universal extension that is trivial on the interior.
More concretely, this means that every additive theory $\BBB\in\QFT(M)$ may be naturally 
decomposed into two distinct pieces of data: (1)~A theory $\AAA\in\QFT(\interior{M})$ on the interior,
which is insensitive to the boundary as postulated by F-locality, and (2)~an ideal $\III\subseteq \ext\AAA$ of its universal
extension that is trivial on the interior, i.e.\ that is only sensitive to the boundary. Specific examples
of such ideals arise from imposing boundary conditions. We shall illustrate
this fact by using the free Klein-Gordon theory as an example. Thus, our results also provide a bridge
between the ideas of  \cite{Sommer} and the concrete constructions in \cite{Casimir,Dappiaggi2,Dappiaggi3,Dappiaggi1}.
\sk

The remainder of this paper is structured as follows:
In Section \ref{sec:geometry} we recall some basic definitions and results
about the causal structure of spacetimes with timelike boundaries, see also \cite{CGS,Solis}.
In Section \ref{sec:categories} we provide a precise definition of the 
categories $\QFT(M)$ and $\QFT(\interior{M})$ by using the ideas of \cite{operad}.
Our universal boundary extension is developed in Section \ref{sec:bdyext},
where we also provide an explicit model in terms of left Kan extension.
Our main results on the characterization of additive quantum field theories
on $M$ are proven in Section \ref{sec:characterization}. Section \ref{sec:KG}
illustrates our construction by focusing on the simple example of the free Klein-Gordon
theory, where more explicit formulas can be developed. 
It is in this context that we provide examples of ideals 
implementing boundary conditions and relate to analytic results, 
e.g.\ \cite{Casimir}. We included Appendix
\ref{app:cattheory} to state some basic definitions and results of category theory 
which will be used in our work.


\section{\label{sec:geometry}Spacetimes with timelike boundary}
We collect some basic facts about spacetimes with timelike boundary,
following \cite[Section 3.1]{Solis} and \cite[Section 2.2]{CGS}. 
For a general introduction to Lorentzian geometry we refer to \cite{BEE, ONeill},
see also \cite[Sections 1.3 and A.5]{BGP} for a concise presentation. 
\sk

We use the term \textit{manifold with boundary} to refer to a Hausdorff, 
second countable, $m$-dimensional smooth manifold  $M$ with boundary,
see e.g.\ \cite{Lee}. This definition subsumes ordinary manifolds 
as manifolds with empty boundary $\partial M =\emptyset$. 
We denote by $\interior{M}\subseteq M$ the submanifold without the boundary.
Every open subset $U\subseteq M$ carries the structure of a manifold with 
(possibly empty) boundary and one has $\interior{U} = U\cap \interior{M}$.
\begin{defi}
A \textit{Lorentzian manifold with boundary} is a 
manifold with boundary that is equipped with a Lorentzian metric. 
\end{defi}
\begin{defi}\label{def:cauchydevelopment}
Let $M$ be a time-oriented Lorentzian manifold with boundary.
The \textit{Cauchy development} $D(S)\subseteq M$ of a subset $S\subseteq M$ 
is the set of points $p\in M$ such that every inextensible (piecewise smooth) 
future directed causal curve stemming from $p$ meets $S$.
\end{defi}

The following properties follow easily from the definition of Cauchy development.
\begin{propo}\label{propo:Cauchy1}
Let $S, S^\prime\subseteq M$ be subsets of a time-oriented Lorentzian manifold $M$
with boundary. Then the following holds true:
\begin{enumerate}
\item[(a)] $S\subseteq S^\prime$ implies $D(S) \subseteq D(S^\prime)$;
\item[(b)] $S \subseteq D(S) = D(D(S))$;
\item[(c)] $D(D(S) \cap D(S^\prime)) = D(S) \cap D(S^\prime)$. 
\end{enumerate}
\end{propo}

We denote by $J_M^\pm(S)\subseteq M$
the \textit{causal future/past} of a subset $S\subseteq M$, 
i.e.\ the set of points that can be reached by a future/past directed 
causal curve stemming from $S$. Furthermore, we denote
by $I_M^\pm(S)\subseteq M$
the \textit{chronological future/past} of a subset $S\subseteq M$, 
i.e.\ the set of points that can be reached by a future/past directed 
timelike curve stemming from $S$. 
\begin{defi}\label{def:causallydisjoint}
Let $M$ be a time-oriented Lorentzian manifold with boundary.
We say that a subset $S\subseteq M$ 
is \textit{causally convex} in $M$ if $J_M^+(S) \cap J_M^-(S) \subseteq S$. 
We say that two subsets $S, S^\prime \subseteq M$ are 
\textit{causally disjoint} in $M$ if $(J_M^+(S) \cup J_M^-(S)) \cap S^\prime = \emptyset$. 
\end{defi}
The following properties are simple consequences of these definitions.
\begin{propo}\label{propo:Cauchy2}
Let $S, S^\prime \subseteq M$ be two subsets of a time-oriented Lorentzian 
manifold  $M$ with boundary. Then the following holds true:
\begin{enumerate}
\item[(a)] $D(S)$ and $D(S^\prime)$ are causally disjoint if 
and only if $S$ and $S^\prime$ are causally disjoint;
\item[(b)] Suppose $S$ and $S^\prime$ are causally disjoint. 
Then the disjoint union $S \sqcup S^\prime \subseteq M$ is causally convex 
if and only if both $S$ and $S^\prime$ are causally convex. 
\end{enumerate}
\end{propo}

The following two definitions play an essential role in our work.
\begin{defi}\label{def:spacetime0}
A \textit{spacetime with timelike boundary} is 
an oriented and time-oriented Lorentzian manifold $M$ with boundary,
such that the pullback of the Lorentzian metric along the boundary inclusion
$\partial M \hookrightarrow M$ defines a Lorentzian metric on the boundary $\partial M$.
\end{defi}

\begin{defi}\label{def:spacetime}
Let $M$ be a spacetime with timelike boundary.
\begin{itemize}
\item[(i)] $\Reg_M$ denotes the category whose objects are causally convex open subsets 
$U\subseteq M$ and whose morphisms $i: U \rightarrow U^\prime$ are inclusions $U \subseteq U^\prime \subseteq M$. 
We call  it the  \textit{category of regions} in $M$.

\item[(ii)] $\Cauchy_M \subseteq \mathrm{Mor}\,\Reg_M$ is the subset of \textit{Cauchy morphisms}
in $\Reg_M$, i.e.\ inclusions $i: U \rightarrow U^\prime$ such that $D(U) = D(U^\prime)$. 

\item[(iii)]  $\Reg_{\interior{M}} \subseteq \Reg_{M}$  is the full subcategory
whose objects are contained in the interior $\interior M$.
We denote by $\Cauchy_{\interior{M}}\subseteq \Cauchy_M$ the Cauchy morphisms 
between objects of $\Reg_{\interior{M}}$. 
\end{itemize}
\end{defi}

\begin{propo}\label{propo:spacetime}
Let $M$ be a spacetime with timelike boundary. For each subset 
$S \subseteq M$ and each object $U \in \Reg_{M}$, 
i.e.\ a causally convex open subset $U \subseteq M$, the following holds true:
\begin{enumerate}
\item[(a)] $I_M^\pm(S)$ is the largest open subset of $J_M^\pm(S)$;
\item[(b)] $J_M^\pm(I_M^\pm(S)) = I_M^\pm(S) = I_M^\pm(J_M^\pm(S))$;
\item[(c)] $S \subseteq \interior M$ implies $D(S) \subseteq \interior M$;
\item[(d)] $D(U) \subseteq M$ is causally convex and open, i.e.\ $D(U)\in\Reg_M$. 
\end{enumerate} 
\end{propo}
\begin{proof}
(a) \& (b): These are standard results in the case of empty boundary,
see e.g.\ \cite{BEE, ONeill, BGP}. The extension to spacetimes with
non-empty timelike boundary can be found in \cite[Section 3.1.1]{Solis}. 
\sk

(c): We show that if $D(S)$ contains a boundary point, so does $S$: 
Suppose $p \in D(S)$ belongs to the boundary of $M$. 
By Definition \ref{def:spacetime0}, the boundary $\partial M$ 
of $M$ can be regarded as a time-oriented Lorentzian manifold with empty boundary, 
hence we can consider a future directed inextensible causal curve $\gamma$ 
in $\partial M$ stemming from $p$. 
Since $\partial M$ is a closed subset of $M$, $\gamma$ must be inextensible 
also as a causal curve in $M$, hence $\gamma$ meets $S$ 
because it stems from $p \in D(S)$. 
Since $\gamma$ lies in $\partial M$ by construction, 
we conclude that $S$ contains a boundary point of $M$. 
\sk

(d): $D(U)\subseteq M$ is causally convex by the definition 
of Cauchy development and by the causal convexity of $U\subseteq M$. 
To check that $D(U)\subseteq M$ is open, we use quasi-limits as in
\cite[Definition 14.7 and Proposition 14.8]{ONeill}: 
First, observe that $I_M^\pm(U)\subseteq M$ is open by (a). Hence,
it is the same to check whether a subset of $I := I_M^+(U) \cup I_M^-(U)$ 
is open in $M$ or in $I$ (with the induced topology). 
Indeed, $U \subseteq D(U) \subseteq I$ because $U$ is open in $M$. 
From now on, in place of $M$, let us therefore consider $I$, equipped with the induced metric, orientation and 
time-orientation. 
By contradiction, assume that there exists $p \in D(U) \setminus U$ 
such that all of its neighborhoods intersect the complement of $D(U)$. 
Then there exists a sequence $\{\alpha_n\}$ of inextensible causal curves in $I$ 
never meeting $U$ such that $\{\alpha_n(0)\}$ converges to $p$. 
We fix a convex cover of $I$ refining the open cover $\{I_M^+(U), I_M^-(U)\}$. 
Relative to the fixed convex cover, the construction of 
quasi-limits allows us to obtain from $\{\alpha_n\}$ 
an inextensible causal curve $\lambda$ through $p \in D(U)$. 
Hence, $\lambda$ meets $U$, say in $q$. 
By the construction of a quasi-limit, $q$ lies on a causal geodesic segment 
between $p_k$ and $p_{k+1}$, two successive limit points 
for $\{\alpha_n\}$ contained in some element of the fixed convex cover. 
It follows that either $p_k$ or $p_{k+1}$ belongs to 
$J_I^+(U) \cap J_I^-(U)$, which is contained in $U$ by causal convexity. 
Hence, we found a subsequence $\{\alpha_{n_j}\}$ of $\{\alpha_n\}$ 
and a sequence of parameters $\{s_j\}$ such that $\{\alpha_{n_j}(s_j)\}$ 
converges to a point of $U$ (either $p_k$ or $p_{k+1}$). 
By construction the sequence $\{\alpha_{n_j}(s_j)\}$ 
is contained in $I \setminus U$, however its limit lies in $U$. 
This contradicts the hypothesis that $U$ is open in $I$. 
\end{proof}

The causal structure of a spacetime $M$ with timelike boundary 
can be affected by several pathologies,
such as the presence of closed future directed causal curves. 
It is crucial to avoid these issues in order to obtain concrete examples of 
our constructions in Section \ref{sec:KG}. 
The following definition is due to \cite[Section 2.2]{CGS} and \cite[Section 3.1.2]{Solis}. 
\begin{defi}\label{def:globhyp}
A spacetime $M$ with timelike boundary is called 
\textit{globally hyperbolic} if the following two properties hold true: 
\begin{itemize}
\item[(i)] {\it Strong causality:} 
Every open neighborhood of each point $p \in M$ 
contains a \textit{causally convex} open neighborhood of $p$. 
\item[(ii)] {\it Compact double-cones:} $J_M^+(p) \cap J_M^-(q)$ 
is compact for all $p,q \in M$. 
\end{itemize}
\end{defi}

\begin{rem}
In the case of empty boundary, this definition agrees with the usual one
in  \cite{BEE, ONeill, BGP}. 
Simple examples of globally hyperbolic spacetimes with non-empty timelike boundary 
are the half space $\{x^{m-1} \geq 0\} \subseteq \bbR^m$, 
the spatial slab $\{0 \leq x^{m-1} \leq 1\} \subseteq \bbR^m$ and the cylinder 
$\{(x^1)^2 + \ldots + (x^{m-1})^2 \leq 1\} \subseteq \bbR^m$ in Minkowski spacetime $\bbR^{m}$,
for $m\geq 2$, as well as all causally convex open subsets thereof. 
\end{rem}

The following results follow immediately
from Definition \ref{def:globhyp} and Proposition \ref{propo:spacetime}. 
\begin{propo}\label{propo:globhyp}
Let $M$ be a globally hyperbolic spacetime with timelike boundary. 
\begin{itemize}
\item[(a)] $M$ admits a cover by causally convex open subsets. 
\item[(b)] For each $U\in \Reg_M$, i.e.\ a causally convex open subset 
$U\subseteq M$, both $U$ and $D(U)$ are globally hyperbolic spacetimes with (possibly empty) 
boundary when equipped with the metric, orientation and time-orientation induced by $M$. 
If moreover $U \subseteq \interior{M}$ is contained in the interior, 
then both $U$ and $D(U)$ have empty boundary. 
\end{itemize}
\end{propo}


\section{\label{sec:categories}Categories of algebraic quantum field theories}
Let $M$ be a spacetime with timelike boundary. 
(In this section we do not have to assume that  $M$ is globally 
hyperbolic in the sense of Definition \ref{def:globhyp}.)
Recall the category $\Reg_M$ of open and causally convex regions in $M$ 
and the subset $\Cauchy_{M}$ of Cauchy morphisms (cf.\ Definition \ref{def:spacetime}). 
Together with our notion of causal disjointness from Definition \ref{def:causallydisjoint}, 
these data provide the geometrical input for 
the traditional definition of algebraic quantum field theories on $M$.
\begin{defi}\label{def:QFT}
An \textit{algebraic quantum field theory} on $M$ is a functor 
\begin{flalign}
\AAA: \Reg_M \longrightarrow \Alg
\end{flalign}
with values in the category $\Alg$ of associative and unital $\ast$-algebras over $\bbC$,
which satisfies the following properties: 
\begin{enumerate}
\item[(i)] \textit{Causality axiom:} For all causally disjoint inclusions 
$i_1: U_1 \rightarrow U \leftarrow U_2: i_2$, 
the induced commutator
\begin{flalign}
\big[\AAA(i_1)(-),\AAA(i_2)(-) \big]_{\AAA(U)}^{} \,:\, \AAA(U_1)\otimes \AAA(U_2)\longrightarrow \AAA(U)
\end{flalign}
is zero.

\item[(ii)] \textit{Time-slice axiom:} 
For all Cauchy morphisms $(i: U \to U^\prime)\in \Cauchy_M$, the map
\begin{flalign}
\AAA(i): \AAA(U) \longrightarrow \AAA(U^\prime)
\end{flalign}
is an $\Alg$-isomorphism.
\end{enumerate} 
We denote by $\qft(M)\subseteq \Alg^{\Reg_M}$ the full subcategory
of the category of functors from $\Reg_M$ to $\Alg$ whose
objects are all algebraic quantum field theories on $M$,
i.e.\ functors fulfilling the causality and time-slice axioms.
(Morphisms in this category are all natural transformations.)
\end{defi}

We shall now show that there exists an alternative, but equivalent, description of
the category $\qft(M)$ which will be more convenient for the technical constructions
in our paper. Following \cite[Section 4.1]{operad}, we observe that the time-slice
axiom in Definition \ref{def:QFT} (ii) is equivalent to considering functors
$\BBB : \Reg_M[\Cauchy_M^{-1}] \to \Alg$ that are defined on the localization
of the category $\Reg_M$ at the set of Cauchy morphisms $\Cauchy_M$.
See Definition \ref{def:localization} for the definition of localizations of categories.
By abstract arguments as in \cite[Section 4.6]{operad}, one observes that
the universal property of localizations implies that the category $\qft(M)$
is equivalent to the full subcategory of the functor category $\Alg^{\Reg_M[\Cauchy_M^{-1}]}$
whose objects are all functors $ \BBB : \Reg_M[\Cauchy_M^{-1}] \to \Alg$ that
satisfy the causality axiom for the pushforward orthogonality relation on 
$\Reg_M[\Cauchy_M^{-1}] $. Loosely speaking, this means that the 
time-slice axiom in Definition \ref{def:QFT} (ii) can be hard-coded by
working on the localized category $\Reg_M[\Cauchy_M^{-1}]$ instead of using the
usual category $\Reg_M$ of regions in $M$.
\sk

The aim of the remainder of this section is to provide an explicit model for
the localization functor $\Reg_{M} \to \Reg_{M}[\Cauchy_M^{-1}]$. 
With this model it will become particularly easy to verify the equivalence 
between the two alternative descriptions of the category $\qft(M)$.
Let us denote by
\begin{subequations}\label{eqn:localization}
\begin{flalign}
\Reg_M[\Cauchy_M^{-1}]\subseteq \Reg_M
\end{flalign}
the full subcategory of $\Reg_M$ whose objects $V \subseteq M$ are stable under Cauchy development,
i.e.\ $D(V) = V$ where $D(V)\subseteq M$ denotes the Cauchy development (cf.\ Definition \ref{def:cauchydevelopment}).
In the following we shall always use letters like
$U\subseteq M$ for generic regions in $\Reg_M$ and $V\subseteq M$ 
for regions that are stable under Cauchy development, i.e.\ objects in $\Reg_M[\Cauchy_M^{-1}]$. 
Recall from Definition \ref{def:spacetime} that each object $U \in \Reg_{M}$ 
is a causally convex open subset $U \subseteq M$, 
hence the Cauchy development $D(U)\subseteq M$ is a causally convex 
open subset by Proposition \ref{propo:spacetime}~(d), 
which is stable under Cauchy development by Proposition \ref{propo:Cauchy1}~(b). 
This shows that $D(U)$ is an object of $\Reg_M[\Cauchy_M^{-1}]$. 
Furthermore, a morphism $i: U \to U^\prime$ in $\Reg_M$ 
is an inclusion $U \subseteq U^\prime$, which induces an 
inclusion $D(U) \subseteq D(U^\prime)$ by Proposition \ref{propo:Cauchy1}~(a) 
and hence a morphism $D(i): D(U) \to D(U^\prime)$ 
in $\Reg_M[\Cauchy_M^{-1}]$. We define a functor
\begin{flalign}
D: \Reg_M \longrightarrow \Reg_M[\Cauchy_M^{-1}]
\end{flalign}
by setting on objects and morphisms
\begin{flalign}
U \longmapsto D(U)\quad,\qquad(i: U \to U^\prime) \longmapsto \big( D(i) : D(U) \to D(U^\prime) \big)\quad.
\end{flalign}
\end{subequations}
Furthermore, let us write
\begin{flalign}
I: \Reg_M[\Cauchy_M^{-1}] \longrightarrow \Reg_M
\end{flalign}
for the full subcategory embedding. 
\begin{lem}\label{lem:reflloc}
$D$ and $I$ form an adjunction  (cf.\ Definition \ref{def:adjunction}) 
\begin{flalign}
\xymatrix{
D: \Reg_M \ar@<0.5ex>[r] & \Reg_M[\Cauchy_M^{-1}]: I \ar@<0.5ex>[l] \quad,
}
\end{flalign}
whose counit is a natural isomorphism (in fact, the identity), 
hence $\Reg_M[\Cauchy_M^{-1}]$ 
is a full reflective subcategory of $\Reg_M$. 
Furthermore, the components of the unit are Cauchy morphisms. 
\end{lem}
\begin{proof}
For $U \in \Reg_M$, the $U$-component of the unit 
\begin{flalign}
\eta: \id_{\Reg_M} \longrightarrow I\, D
\end{flalign}
is given by the inclusion $U \subseteq D(U)$ of $U$ 
into its Cauchy development, which is a Cauchy morphism, 
see Proposition \ref{propo:Cauchy1}~(b) and Definition \ref{def:spacetime}~(ii). 
For $V \in \Reg_M[\Cauchy_M^{-1}]$, the $V$-component of the counit 
\begin{flalign}
\epsilon: D\, I \longrightarrow \id_{\Reg_M[\Cauchy_M^{-1}]} 
\end{flalign}
is given by the identity of the object $D(V) = V$.
The triangle identities hold trivially.
\end{proof}

\begin{propo}\label{propo:locmodel}
The category $\Reg_M[\Cauchy_M^{-1}]$ and the functor 
$D: \Reg_M \to \Reg_M[\Cauchy_M^{-1}]$ defined in \eqref{eqn:localization} 
provide a model for the localization of $\Reg_M$ at $\Cauchy_M$.
\end{propo}
\begin{proof}
We have to check all the requirements listed in Definition \ref{def:localization}.
\sk

\noindent (a)~By Definition \ref{def:spacetime}, 
for each Cauchy morphism $i: U \to U^\prime$ one has $D(U) = D(U^\prime)$ 
and hence $D(i) = \id_{D(U)}$ is an isomorphism in $\Reg_M[\Cauchy_M^{-1}]$. 
\sk

\noindent (b)~Let $F: \Reg_M \to \DD$ be any functor to a category $\DD$ that sends morphisms
in $\Cauchy_M$ to $\DD$-isomorphisms. Using Lemma \ref{lem:reflloc}, we 
define $F_W := F \, I : \Reg_M[\Cauchy_M^{-1}]\to \DD$ and consider the natural transformation
$F \eta :F \to F\,I\,D = F_W\,D $ obtained by the unit of the adjunction $D\dashv I$.
Because all components of $\eta$ are Cauchy morphisms (cf.\ Lemma \ref{lem:reflloc}),
$F\eta$ is a natural isomorphism.
\sk

\noindent (c)~Let $G,H : \Reg_M[\Cauchy_M^{-1}]\to \DD$ be two functors. We have to show that
the map
\begin{flalign}\label{eqn:homfullyfaithfulmap}
\Hom_{\DD^{\Reg_M[\Cauchy_M^{-1}]}}\big(G,H\big)\longrightarrow \Hom_{\DD^{\Reg_M}}\big(G D,H D\big)
\end{flalign}
is a bijection. Let us first prove injectivity: Let  $\xi, \widetilde{\xi}: G \to H$ 
be two natural transformations such that $\xi D = \widetilde{\xi} D$. 
Using Lemma \ref{lem:reflloc}, we obtain commutative diagrams
\begin{flalign}
\xymatrix{
\ar[d]_-{G\epsilon } G D I \ar[r]^-{\xi D I} & H D I\ar[d]^-{H\epsilon}\\
G\ar[r]_-{\xi}& H
}\qquad\qquad
\xymatrix{
\ar[d]_-{G\epsilon } G D I \ar[r]^-{\widetilde{\xi} D I} & H D I\ar[d]^-{H\epsilon}\\
G\ar[r]_-{\widetilde{\xi}}& H
}
\end{flalign}
where the vertical arrows are natural isomorphisms because the counit $\epsilon$ is an isomorphism.
Recalling that by hypothesis $\xi D = \widetilde{\xi} D$, it follows that $\xi = \widetilde{\xi}$.
Hence, the map \eqref{eqn:homfullyfaithfulmap} is injective.
\sk

It remains to prove that \eqref{eqn:homfullyfaithfulmap} is also surjective.
Let $\chi : G D\to H D$ be any natural transformation. Using Lemma
\ref{lem:reflloc}, we obtain a commutative diagram
\begin{flalign}
\xymatrix{
\ar[d]_-{G D\eta} G D \ar[r]^-{\chi} & H D\ar[d]^-{H D\eta}\\
GDID\ar[r]_-{\chi ID}& HDID
}
\end{flalign}
where the vertical arrows are natural isomorphisms because the components of the unit $\eta$ 
are Cauchy morphisms and $D$ assigns isomorphisms to them. Let us define a
natural transformation $\xi : G\to H$ by the commutative diagram
\begin{flalign}
\xymatrix{
\ar[d]_-{G\epsilon^{-1}} G \ar[r]^-{\xi} & H\\
GDI\ar[r]_-{\chi I}& HDI\ar[u]_-{H\epsilon}
}
\end{flalign}
where we use that $\epsilon$ is a natural isomorphism (cf.\ Lemma \ref{lem:reflloc}).
Combining the last two diagrams, one easily computes that $\xi D = \chi$ by using also the triangle identities
of the adjunction $D \dashv I$. Hence, the map \eqref{eqn:homfullyfaithfulmap} is surjective.
\end{proof}

We note that there exist two (a priori different) options to define
an orthogonality relation on the localized category $\Reg_M[\Cauchy_M^{-1}]$, 
both of which are provided by \cite[Lemma 4.29]{operad}: 
(1)~The pullback orthogonality relation along the full subcategory embedding 
$I : \Reg_M[\Cauchy_M^{-1}] \to \Reg_M$ and
(2)~the pushforward orthogonality relation along the localization functor 
$D: \Reg_M \to \Reg_M[\Cauchy_M^{-1}]$. 
In our present scenario, both constructions coincide and
one concludes that two $\Reg_M[\Cauchy_M^{-1}]$-morphisms are 
orthogonal precisely when they are orthogonal in $\Reg_M$. 
Summing up, we obtain
\begin{lem}\label{lem:orthloc}
We say that two morphisms in the full subcategory 
$\Reg_M[\Cauchy_M^{-1}] \subseteq \Reg_M$ 
are orthogonal precisely when they are orthogonal in $\Reg_M$ 
(i.e.\ causally disjoint, cf.\ Definition \ref{def:causallydisjoint}). 
Then both functors $D: \Reg_M \to \Reg_M[\Cauchy_M^{-1}]$ and 
$I: \Reg_M[\Cauchy_M^{-1}] \to \Reg_M$ 
preserve (and also detect) the orthogonality relations. 
\end{lem}
\begin{proof}
For $I$ this holds trivially, while for $D$ see Proposition \ref{propo:Cauchy2}~(a). 
\end{proof}

With these preparations we may now define our alternative description
of the category of algebraic quantum field theories.
\begin{defi}\label{def:QFT2}
We denote by $\QFT(M)\subseteq \Alg^{\Reg_M[\Cauchy_M^{-1}]}$ 
the full subcategory whose objects are all
functors $\BBB : \Reg_M[\Cauchy_M^{-1}] \to \Alg$ 
that satisfy the following version of the \textit{causality axiom}:
For all causally disjoint inclusions $i_1 : V_1\to V \leftarrow V_2 : i_2$
in $\Reg_M[\Cauchy_M^{-1}]$, i.e.\ $V$, $V_1$ and $V_2$ are stable under Cauchy development,
the induced commutator
\begin{flalign}
\big[\BBB(i_1)(-),\BBB(i_2)(-) \big]_{\BBB(V)}^{} \,:\, \BBB(V_1)\otimes \BBB(V_2)\longrightarrow \BBB(V)
\end{flalign}
is zero.
\end{defi}
\begin{theo}\label{theo:qftQFTequivalence}
By pullback, the adjunction $D \dashv I$ of Lemma \ref{lem:reflloc}
induces an adjoint equivalence (cf.\ Definition \ref{def:adjointequivalence})
\begin{flalign}\label{eqn:qftQFTequivalence}
\xymatrix{
I^\ast: \qft(M) \ar@<0.75ex>[r]_-{\sim} & \QFT(M): D^\ast \ar@<0.75ex>[l]
}\quad.
\end{flalign}
In particular, the two categories
$\qft(M)$ of Definition \ref{def:QFT} and $\QFT(M)$ of Definition \ref{def:QFT2}
are equivalent.
\end{theo}
\begin{proof}
It is trivial to check that the adjunction 
$D:  \Reg_M \rightleftarrows \Reg_M[\Cauchy_M^{-1}] : I$ 
induces an adjunction 
\begin{flalign}
\xymatrix{
I^\ast: \Alg^{\Reg_M} \ar@<0.5ex>[r] 
& \Alg^{\Reg_M[\Cauchy_M^{-1}]}: D^\ast \ar@<0.5ex>[l]
}
\end{flalign}
between functor categories. Explicitly, the unit 
$\widetilde{\eta} : \id_{\Alg^{\Reg_M}} \to D^\ast\, I^\ast$
has components
\begin{flalign}
\widetilde{\eta}_\AAA^{} := \AAA\eta \, :\,\AAA \longrightarrow D^\ast (I^\ast(\AAA)) = \AAA ID  \quad,
\end{flalign}
where $\AAA: \Reg_M \to \Alg$ is any functor 
and $\eta: \id_{\Reg_M} \to I\, D$ denotes the unit of $D \dashv I$.
The counit 
$\widetilde{\epsilon} :  I^\ast\, D^\ast \to \id_{\Alg^{\Reg_M[\Cauchy_M^{-1}]}}$
has components
\begin{flalign}
\widetilde{\epsilon}_\BBB^{} := \BBB\epsilon  \,:\, I^\ast(D^\ast(\BBB)) = \BBB DI  \longrightarrow \BBB \quad, 
\end{flalign}
where $\BBB : \Reg_M[\Cauchy_M^{-1}] \to \Alg$ is any functor 
and $\epsilon: D\, I \to \id_{\Reg_M[\Cauchy_M^{-1}]}$ 
denotes the counit of $D \dashv I$. 
The triangle identities for $I^\ast \dashv D^\ast$ 
follow directly from those of $D \dashv I$. 
\sk

Next, we have to prove that this adjunction restricts to  the claimed
source and target categories in \eqref{eqn:qftQFTequivalence}. 
Given $\AAA\in \qft(M)\subseteq \Alg^{\Reg_M}$,
the functor $I^\ast (\AAA) = \AAA \, I : \Reg_{M}[\Cauchy_M^{-1}]\to \Alg$
satisfies the causality axiom of Definition \ref{def:QFT2} because of Lemma \ref{lem:orthloc}.
Hence, $I^\ast (\AAA) \in\QFT(M)$. Vice versa, given $\BBB\in \QFT(M)\subseteq 
\Alg^{\Reg_M[\Cauchy_M^{-1}]}$, the functor $D^\ast (\BBB) = \BBB\,D :\Reg_M \to \Alg$
satisfies the causality axiom of Definition \ref{def:QFT2} because
of Lemma \ref{lem:orthloc} and the time-slice axiom of Definition \ref{def:QFT2}
because $D$ sends by construction morphisms in $\Cauchy_M$ to isomorphisms.
Hence, $D^\ast (\BBB)\in\qft(M)$.
\sk

Using Lemma \ref{lem:reflloc}, we obtain that the counit $\widetilde{\epsilon}$ 
of the restricted adjunction \eqref{eqn:qftQFTequivalence} is an isomorphism. 
Furthermore, all components of $\eta$ are Cauchy morphisms, 
hence $\widetilde{\eta}_\AAA^{} = \AAA\eta$ is an isomorphism for all 
$\AAA\in \qft(M)$, i.e.\ the unit $\widetilde{\eta}$ is an isomorphism. 
This completes the proof that \eqref{eqn:qftQFTequivalence}  is an adjoint equivalence.
\end{proof}

\begin{rem}\label{rem:qftQFTequivalence}
Theorem \ref{theo:qftQFTequivalence} provides us with a constructive
prescription of how to change between the two equivalent formulations of
algebraic quantum field theories given in Definitions \ref{def:QFT} and \ref{def:QFT2}.
Concretely, given any $\AAA \in \qft(M)$, i.e.\ a functor $\AAA : \Reg_M\to \Alg$ satisfying the 
causality and time-slice axioms as in Definition \ref{def:QFT},  the corresponding
quantum field theory $I^\ast(\AAA)\in\QFT(M)$ in the sense of Definition \ref{def:QFT2}
reads as follows: It is the functor $I^\ast(\AAA) = \AAA\,I : \Reg_M[\Cauchy_M^{-1}]\to\Alg$
on the category of regions $V\subseteq M$ that are stable under Cauchy development,
which assigns to $V\subseteq M$ the algebra $\AAA(V)\in\Alg$ and to an inclusion
$i : V\to V^\prime$ the algebra map $\AAA(i) : \AAA(V)\to \AAA(V^\prime)$.
More interestingly, given $\BBB\in \QFT(M)$, i.e.\ a functor $\BBB : \Reg_M[\Cauchy_M^{-1}]\to \Alg$
satisfying the causality axiom as in Definition \ref{def:QFT2}, the corresponding
quantum field theory $D^\ast(\BBB)\in\qft(M)$ in the sense of Definition \ref{def:QFT}
reads as follows: It is the functor $D^\ast(\BBB) = \BBB\,D : \Reg_M\to \Alg$ defined on the
category of (not necessarily Cauchy development stable) regions $U\subseteq M$,
which assigns to $U\subseteq M$ the algebra $\BBB(D(U))$ corresponding to the Cauchy development of $U$ 
and to an inclusion $i : U\to U^\prime$ the algebra map $\BBB(D(i)) : \BBB(D(U))\to \BBB(D(U^\prime))$
associated to the inclusion $D(i) :D(U)\to D(U^\prime)$ of Cauchy developments.
\end{rem}

\begin{rem}\label{rem:alsoint}
It is straightforward to check that 
the results of this section still hold true when one replaces 
$\Reg_{M}$ with its full subcategory $\Reg_{\interior{M}}$ 
of regions contained in the interior of $M$ and $\Cauchy_M$ with $\Cauchy_{\interior{M}}$
(cf.\ Definition \ref{def:spacetime}).
This follows from the observation that the Cauchy development 
of a subset of the interior of $M$ is also contained in $\interior{M}$, 
as shown in Proposition \ref{propo:spacetime}~(c). 
We denote by 
\begin{flalign}
\QFT(\interior{M}) \subseteq \Alg^{\Reg_{\interior{M}}[\Cauchy_{\interior{M}}^{-1}]}
\end{flalign}
the category of algebraic quantum field theories in the sense of Definition \ref{def:QFT2}
on the interior regions of $M$. Concretely, an object $\AAA \in \QFT(\interior{M})$
is a functor $\AAA : \Reg_{\interior{M}}[\Cauchy_{\interior{M}}^{-1}] \to \Alg$
that satisfies the causality axiom of Definition \ref{def:QFT2} for causally disjoint interior
regions.
\end{rem}


\section{\label{sec:bdyext}Universal boundary extension}
The goal of this section is to develop a universal 
construction to extend quantum field theories from the interior of 
a spacetime $M$ with timelike boundary to the whole spacetime. 
(Again, we do not have to assume that  $M$ is
globally hyperbolic in the sense of Definition \ref{def:globhyp}.)
Loosely speaking, our extended quantum field theory will have the
following pleasant properties: (1)~It describes precisely those observables 
that are generated from the original theory on the interior, (2)~it does not
require a choice of boundary conditions, (3)~specific choices of boundary 
conditions correspond to ideals of our extended quantum field theory.
We also refer to Section \ref{sec:characterization} for more details
on the properties~(1) and~(3).
\sk

The starting point for this construction  is
the full subcategory inclusion $\Reg_{\interior{M}}\subseteq \Reg_{M}$
defined by selecting only the regions of $\Reg_{M}$ 
that lie in the interior of $M$ (cf.\ Definition \ref{def:spacetime}). 
We denote the corresponding embedding functor by 
\begin{flalign}\label{eqn:j}
j: \Reg_{\interior{M}} \longrightarrow \Reg_M
\end{flalign}
and notice that $j$ preserves (and also detects) causally disjoint inclusions, 
i.e.\ $j$ is a full orthogonal subcategory embedding in the terminology 
of \cite{operad}. Making use of Proposition \ref{propo:locmodel},
Lemma \ref{lem:reflloc} and Remark \ref{rem:alsoint}, we define a functor 
$J: \Reg_{\interior{M}}[\Cauchy_{\interior{M}}^{-1}] 
\to \Reg_M[\Cauchy_M^{-1}]$ on the localized categories
via the commutative diagram
\begin{subequations}\label{eqn:J}
\begin{flalign}
\xymatrix{
\ar[d]_-{I}\Reg_{\interior{M}}[\Cauchy_{\interior{M}}^{-1}]  \ar[r]^-{J} & \Reg_M[\Cauchy_M^{-1}]\\
\Reg_{\interior{M}} \ar[r]_-{j}&\Reg_M\ar[u]_-{D}
}
\end{flalign}
Notice that $J$ is simply an embedding functor, which
acts on objects and morphisms as
\begin{flalign}
V\subseteq \interior{M} \longmapsto V\subseteq M\quad,\qquad
(i : V\to V^\prime) \longmapsto (i: V\to V^\prime)\quad.
\end{flalign}
\end{subequations}
From this explicit description it is clear that $J$ preserves (and also detects) causally disjoint inclusions,
i.e.\ it is a full orthogonal subcategory embedding. The constructions in 
\cite[Section 5.3]{operad} (see also \cite{involution} for details how to treat $\ast$-algebras)
then imply that $J$ induces an adjunction
\begin{flalign}\label{eqn:EXTRES}
\xymatrix{
\ext: \QFT(\interior{M}) \ar@<0.5ex>[r] 
& \QFT(M): \res \ar@<0.5ex>[l]
} 
\end{flalign}
between the category of quantum field theories on the interior $\interior{M}$  (cf.\ Remark \ref{rem:alsoint})
and the category of quantum field theories on the whole spacetime $M$.
The right adjoint $\res := J^\ast : \QFT(M) \to  \QFT(\interior{M}) $ is the pullback
along $J : \Reg_{\interior{M}}[\Cauchy_{\interior{M}}^{-1}]  \to  \Reg_M[\Cauchy_M^{-1}]$,
i.e.\ it restricts quantum field theories defined on $M$ to the interior $\interior{M}$.
The left adjoint $\ext : \QFT(\interior{M}) \to  \QFT(M) $ should be regarded as a universal extension
functor which extends quantum field theories on the interior $\interior{M}$ to the whole
spacetime $M$. The goal of this section is to analyze the properties of this extension functor
and to develop an explicit model that allows us to do computations in the sections below.
\sk

An important structural result, whose physical relevance is explained in Remark 
\ref{rem:Flocality} below, is the following proposition.
\begin{propo}\label{propo:EXTRESunit}
The unit 
\begin{flalign}
\eta: \id_{\QFT(\interior{M})} \longrightarrow \res\, \ext 
\end{flalign}
of the adjunction \eqref{eqn:EXTRES} is a natural isomorphism. 
\end{propo}
\begin{proof}
This is a direct consequence of the fact that the functor $J$ given in \eqref{eqn:J} 
is a full orthogonal subcategory embedding and the general result in \cite[Proposition 5.6]{operad}. 
\end{proof}

\begin{rem}\label{rem:Flocality}
The physical interpretation of this result is as follows: 
Let $\AAA \in \QFT(\interior{M})$ be a quantum field 
theory defined only on the interior $\interior{M}$ of $M$ 
and let $\BBB := \ext\, \AAA \in \QFT(M)$ denote its universal 
extension to the whole spacetime $M$. The $\AAA$-component 
\begin{flalign}
\eta_{\AAA}: \AAA \longrightarrow \res\, \ext\, \AAA
\end{flalign}
of the unit of the adjunction \eqref{eqn:EXTRES} 
allows us to compare $\AAA$ with the restriction $\res\, \BBB$ 
of its extension $\BBB = \ext\, \AAA$. 
Since $\eta_{\AAA}$ is an isomorphism by Proposition \ref{propo:EXTRESunit}, 
restricting the extension $\BBB$ recovers our original theory $\AAA$ up to isomorphism. 
This allows us to interpret the left adjoint 
$\ext : \QFT(\interior{M}) \to \QFT(M)$ as a genuine extension prescription. 
Notice that this also proves that the universal extension $\ext\, \AAA \in \QFT(M)$ 
of any theory $\AAA \in \QFT(\interior{M})$ on the interior satisfies F-locality 
\cite{Flocality}.
\end{rem}

We next address the question how to \textit{compute} the  extension functor 
$\ext : \QFT(\interior{M}) \to \QFT(M)$ explicitly. A crucial step towards
reaching this goal is to notice that $\ext$ may be computed by a left Kan extension.
\begin{propo}\label{propo:LanJ}
Consider the adjunction
\begin{flalign}
\xymatrix{
\Lan_J : \Alg^{\Reg_{\interior{M}}[\Cauchy_{\interior{M}}^{-1}]} \ar@<0.5ex>[rr] 
&& \Alg^{\Reg_M[\Cauchy_M^{-1}]}: J^\ast \ar@<0.5ex>[ll]
} \quad,
\end{flalign}
corresponding to left Kan extension along the functor $J : \Reg_{\interior{M}}[\Cauchy_{\interior{M}}^{-1}]
\to \Reg_M[\Cauchy_M^{-1}]$. Then the restriction
of $\Lan_J$ to $\QFT(\interior{M})$ induces a functor
\begin{flalign}\label{eqn:LanJ}
\Lan_J : \QFT(\interior{M})\longrightarrow \QFT(M)
\end{flalign}
that is left adjoint to the restriction functor $\res: \QFT(M)\to  \QFT(\interior{M})$ in \eqref{eqn:EXTRES}.
Due to uniqueness (up to unique natural isomorphism) of adjoint functors (cf.\ Proposition \ref{propo:adjointunique}), it follows
that $\ext\cong \Lan_J$, i.e.\ \eqref{eqn:LanJ} is a model for the extension functor $\ext$ in \eqref{eqn:EXTRES}.
\end{propo}
\begin{proof}
A general version of this problem has been addressed in \cite[Section 6]{operad}.
Using in particular \cite[Corollary 6.5]{operad}, we observe that we can prove this proposition
by showing that every object $V\in \Reg_M[\Cauchy_M^{-1}]$ is $J$-closed in the sense of
\cite[Definition 6.3]{operad}. In our present scenario, this amounts to proving that
for all causally disjoint inclusions $i_1 : V_1\to V \leftarrow V_2 : i_2$ 
with $V_1,V_2 \in \Reg_{\interior{M}}[\Cauchy_{\interior{M}}^{-1}]$ in the interior
and $V \in \Reg_M[\Cauchy_M^{-1}]$ not necessarily in the interior, there exists
a factorization of both $i_1$ and $i_2$ through a common interior region.
Let us consider the Cauchy development $D(V_1 \sqcup V_2)$ 
of the disjoint union and the canonical inclusions 
$j_1: V_1 \to D(V_1 \sqcup V_2) \leftarrow V_2: j_2$. 
As we explain below, $D(V_1 \sqcup V_2) \in \Reg_{\interior{M}}[\Cauchy_{\interior{M}}^{-1}]$ is an interior region 
and $j_1, j_2$ provide the desired factorization: 
Since the open set $V_1 \sqcup V_2 \subseteq \interior M$ 
is causally convex by Proposition \ref{propo:Cauchy2}~(b),
$D(V_1 \sqcup V_2)$ 
is causally convex, open and contained in the interior $\interior M$ 
by Proposition \ref{propo:spacetime}~(c-d). It is, moreover, 
stable under Cauchy development by Proposition \ref{propo:Cauchy1}~(b), 
which also provides the inclusion 
$V_k \subseteq V_1 \sqcup V_2 \subseteq D(V_1 \sqcup V_2)$ 
inducing $j_k$, for $k=1,2$. Consider now the chain of inclusions 
$V_k \subseteq V_1 \sqcup V_2 \subseteq V$ 
corresponding to $i_k$, for $k=1,2$. 
From the stability under Cauchy development of $V_1$, $V_2$ and $V$,
we obtain also the chain of inclusions 
$V_k \subseteq D(V_1 \sqcup V_2) \subseteq V$, for $k=1,2$, 
that exhibits the desired factorization
\begin{flalign}
\xymatrix{
&V&\\
\ar@/^1pc/[ru]^-{i_1} V_1 \ar@{-->}[r]_-{j_1} & D(V_1\sqcup V_2)  \ar@{-->}[u]^-{D(i_1\sqcup i_2)}& \ar@{-->}[l]^-{j_2} \ar@/_1pc/[lu]_-{i_2}V_2
}
\end{flalign}
which completes the proof.
\end{proof}

We shall now briefly review a concrete model for left Kan extension 
along full subcategory embeddings that was developed 
in \cite[Section 6]{operad}. This model is obtained by means of
abstract operadic techniques, but it admits an intuitive graphical interpretation
that we explain in Remark \ref{rem:graphical} below.
It allows us to compute quite explicitly the extension
$\ext \AAA = \Lan_J \AAA \in \QFT(M)$ of a quantum field theory
$\AAA \in \QFT(\interior{M})$ defined on the interior $\interior{M}$ 
to the whole spacetime $M$. The functor 
$\ext \AAA : \Reg_{M}[\Cauchy_M^{-1}]\to \Alg$ describing the 
extended quantum field theory reads as follows:
To $V\in \Reg_{M}[\Cauchy_M^{-1}]$ it assigns a quotient algebra
\begin{flalign}\label{eqn:extAAAV}
\ext\AAA(V) \, =\, \bigoplus_{\und{i}: \und{V}\to V}\!\!\AAA(\und{V}) \bigg/\!\!\!\sim
\end{flalign}
that we will describe now in detail. The direct sum (of vector spaces) in \eqref{eqn:extAAAV} runs over all tuples
$\und{i} :\und{V}\to V$ of morphisms in $\Reg_{M}[\Cauchy_M^{-1}]$,
i.e.\ $\und{i} = (i_1 : V_1 \to V, \cdots, i_n : V_n\to V)$ for some $n\in\bbZ_{\geq 0}$,
with the requirement that all sources $V_k\in \Reg_{\interior{M}}[\Cauchy_{\interior{M}}^{-1}]$
are interior regions. (Notice that the regions $V_k$ are not assumed to be causally disjoint and 
that the empty tuple, i.e.\ $n=0$, is also allowed.) The vector space
$\AAA(\und{V})$ is defined by the tensor product of vector spaces
\begin{flalign}\label{eqn:extAAAVelement}
\AAA(\und{V}) \,:=\,\bigotimes_{k=1}^{\vert\und{V}\vert} \AAA(V_k)\quad,
\end{flalign}
where $\vert \und{V}\vert$ is the length of the tuple. (For the empty tuple, we set $\AAA(\emptyset) = \bbC$.)
This means that the (homogeneous) elements
\begin{flalign}\label{eqn:extAAAVbeforesim}
(\und{i},\und{a}) \in \bigoplus_{\und{i}: \und{V}\to V}\!\!\AAA(\und{V})
\end{flalign}
are given by pairs consisting of a tuple of morphisms $\und{i} :\und{V}\to V$ (with all $V_k$ in the interior)
and an element $\und{a} \in\AAA(\und{V}) $ of the corresponding 
tensor product vector space \eqref{eqn:extAAAVelement}.
The product on \eqref{eqn:extAAAVbeforesim} is given on homogeneous elements by
\begin{flalign}\label{eqn:extAAAVproduct}
(\und{i},\und{a})\,(\und{i}^\prime,\und{a}^\prime) \,:=\, \big((\und{i},\und{i}^\prime), \und{a}\otimes\und{a}^\prime\big) \quad,
\end{flalign}
where $(\und{i},\und{i}^\prime) = (i_1,\dots,i_n , i^\prime_1,\dots,i^\prime_m)$ is the concatenation
of tuples. The unit element in \eqref{eqn:extAAAVbeforesim} is $\oone := (\emptyset,1)$, 
where $\emptyset$ is the empty tuple and $1\in\bbC$, and the $\ast$-involution is defined by
\begin{flalign}\label{eqn:extAAAVinvolution}
\big((i_1,\dots,i_n) , a_1\otimes \cdots\otimes a_n\big)^\ast \,:=\,
\big( (i_n,\dots,i_1), a_n^\ast \otimes\cdots\otimes a_1^\ast\big)
\end{flalign} 
and $\bbC$-antilinear extension.
Finally, the quotient in \eqref{eqn:extAAAV} is by the two-sided $\ast$-ideal
of the algebra \eqref{eqn:extAAAVbeforesim} that is generated by
\begin{flalign}\label{extAAAVrelations}
\Big(\und{i}\big(\und{i}_1,\dots,\und{i}_n\big) , \und{a}_1\otimes\cdots\otimes\und{a}_n\Big)
- \Big(\und{i}, \AAA(\und{i}_1)\big(\und{a}_1\big)\otimes\cdots\otimes\AAA(\und{i}_n)\big(\und{a}_n\big)\Big)
 \in  \bigoplus_{\und{i}: \und{V}\to V}\!\!\AAA(\und{V})\quad,
\end{flalign}
for all tuples $\und{i} : \und{V}\to V$ of length $\vert \und{V}\vert =n \geq 1$ (with all $V_k$ in the interior),
all tuples $\und{i}_k : \und{V}_k \to V_k$ of  $\Reg_{\interior{M}}[\Cauchy_{\interior{M}}^{-1}]$-morphisms 
(possibly of length zero), for $k=1,\dots, n$, and all $\und{a}_k\in \AAA(\und{V}_k)$, for $k=1,\dots,n$.
The tuple in the first term of \eqref{extAAAVrelations} is defined by composition
\begin{subequations}
\begin{flalign}
\und{i}\big(\und{i}_1,\dots,\und{i}_n\big)
\,:=\, \big(i_1\,i_{11},\dots, i_1\,i_{1\vert \und{V}_1\vert}, \dots, i_n\,i_{n1},\dots, i_n\,i_{n\vert\und{V}_n\vert}\big)
\end{flalign}
and the expressions $\AAA(\und{i})\big(\und{a}\big)$ in the second term are determined by
\begin{flalign}
 \AAA(\und{i}) \,:\, \AAA(\und{V})\longrightarrow \AAA(V)~,~~ a_1\otimes \cdots \otimes a_n \longmapsto
\AAA(i_1)\big(a_1\big)\,\cdots\,\AAA(i_n)\big(a_n\big)\quad,
\end{flalign}
\end{subequations}
where multiplication in $\AAA(V)$ is denoted by juxtaposition.
To a morphism $ i^\prime : V\to V^\prime $ in $\Reg_{M}[\Cauchy_M^{-1}]$,
the functor $\ext \AAA : \Reg_{M}[\Cauchy_M^{-1}]\to \Alg$ assigns
the algebra map
\begin{flalign}
\ext\AAA(i^\prime ) \,:\, \ext\AAA(V) \longrightarrow \ext\AAA(V^\prime) ~,~~
[\und{i},\und{a}] \longmapsto  \big[i^\prime(\und{i}) , \und{a}\big]\quad,
\end{flalign}
where we used square brackets to indicate equivalence classes in \eqref{eqn:extAAAV}.
\begin{rem}\label{rem:graphical}
The construction of the algebra $\ext\AAA(V)$ above admits an intuitive graphical interpretation:
We shall visualize the (homogeneous) elements $(\und{i},\und{a})$ in 
\eqref{eqn:extAAAVbeforesim} by decorated trees
\begin{flalign}
\begin{tikzpicture}[cir/.style={circle,draw=black,inner sep=0pt,minimum size=1.5mm},
        poin/.style={circle, inner sep=0pt,minimum size=0mm}]
\node[poin] (Mout) [label=above:{\small $V$}] at (0,1) {};
\node[cir] (Min1) [label=below:{\small $a_1$}] at (-0.5,0) {};
\node[cir] (Min2) [label=below:{\small $a_n$}] at (0.5,0) {};
\node[poin] (V)  at (0,0.6) {};
\draw[thick] (Min1) -- (V);
\draw[thick] (Min2) -- (V);
\draw[thick] (V) -- (Mout);
\node[poin] at (0,0) {{\small $\cdots$}};
\end{tikzpicture}
\end{flalign}
where $a_k \in \AAA(V_k)$ is an element of the algebra $\AAA(V_k)$ associated to
the interior region $V_k\subseteq V$, for all $k=1,\dots,n$. We interpret such a decorated 
tree as a formal product of the formal pushforward along $\und{i} : \und{V}\to V$
of the family of observables $a_k \in \AAA(V_k)$. The product \eqref{eqn:extAAAVproduct}
is given by concatenation of the inputs of the individual decorated trees, i.e.\
\begin{flalign}
\begin{tikzpicture}[cir/.style={circle,draw=black,inner sep=0pt,minimum size=1.5mm},
        poin/.style={circle, inner sep=0pt,minimum size=0mm}]
\node[poin] (Mout) [label=above:{\small $V$}] at (0,1) {};
\node[cir] (Min1) [label=below:{\small $a_1$}] at (-0.5,0) {};
\node[cir] (Min2) [label=below:{\small $a_n$}] at (0.5,0) {};
\node[poin] (V)  at (0,0.6) {};
\draw[thick] (Min1) -- (V);
\draw[thick] (Min2) -- (V);
\draw[thick] (V) -- (Mout);
\node[poin] at (0,0) {{\small $\cdots$}};
\node[poin] (MMout) [label=above:{\small $V$}] at (2,1) {};
\node[cir] (MMin1) [label=below:{\small $a^\prime_1$}] at (1.5,0) {};
\node[cir] (MMin2) [label=below:{\small $a^\prime_m$}] at (2.5,0) {};
\node[poin] (VV)  at (2,0.6) {};
\draw[thick] (MMin1) -- (VV);
\draw[thick] (MMin2) -- (VV);
\draw[thick] (VV) -- (MMout);
\node[poin] at (2,0) {{\small $\cdots$}};
\node[poin] (MMMout) [label=above:{\small $V$}] at (5,1) {};
\node[cir] (MMMin1) [label=below:{\small $a_1$}] at (4.5,0) {};
\node[cir] (MMMin2) [label=below:{\small $a^\prime_m$}] at (5.5,0) {};
\node[poin] (VVV)  at (5,0.6) {};
\draw[thick] (MMMin1) -- (VVV);
\draw[thick] (MMMin2) -- (VVV);
\draw[thick] (VVV) -- (MMMout);
\node[poin] at (5,0) {{\small $\cdots$}};
\node[poin] at (3.5,0.5) {{$=$}};
\node[poin] at (1,0.5) {{$\cdot$}};
\end{tikzpicture}
\end{flalign}
where the decorated tree on the right-hand side has $n+m$ inputs.
The $\ast$-involution \eqref{eqn:extAAAVinvolution} may be visualized
by reversing the input profile and applying $\ast$ to each element $a_k\in\AAA(V_k)$.
Finally, the $\ast$-ideal in \eqref{extAAAVrelations} implements the following relations: 
Assume that $(\und{i},\und{a})$ is such that the sub-family of embeddings
$(i_{k},i_{k+1},\dots, i_{l}) : (V_{k},V_{k+1},\dots, V_{l})\to V$ factorizes through some common
interior region, say $V^\prime \subseteq V$. Using the original functor $\AAA \in \QFT(\interior{M})$,
we may form the product $\AAA(i_{k}) (a_{k})\,\cdots\,\AAA(i_{l})(a_{l})$ in the algebra $\AAA(V^\prime)$,
which we denote for simplicity  by $a_{k}\cdots a_l\in  \AAA(V^\prime)$. We then have the relation
\begin{flalign}
\begin{tikzpicture}[cir/.style={circle,draw=black,inner sep=0pt,minimum size=1.5mm},
        poin/.style={circle, inner sep=0pt,minimum size=0mm}]
\node[poin] (Mout) [label=above:{\small $V$}] at (0,1.5) {};
\node[cir] (Min1) [label=below:{\small $a_1$}] at (-1.5,-0.5) {};
\node[cir] (Min2) [label=below:{\small $a_k$}] at (-0.5,-0.5) {};
\node[cir] (Min3) [label=below:{\small $a_l$}] at (0.5,-0.5) {};
\node[cir] (Min4) [label=below:{\small $a_n$}] at (1.5,-0.5) {};
\node[poin] (V)  at (0,0.6) {};
\draw[thick] (Min1) -- (V);
\draw[thick] (Min2) -- (V);
\draw[thick] (Min3) -- (V);
\draw[thick] (Min4) -- (V);
\draw[thick] (V) -- (Mout);
\node[poin] at (-1,-0.5) {{\small $\cdots$}};
\node[poin] at (0,-0.5) {{\small $\cdots$}};
\node[poin] at (1,-0.5) {{\small $\cdots$}};
\node[poin] (MMout) [label=above:{\small $V$}] at (5,1.5) {};
\node[cir] (MMin1) [label=below:{\small $a_1$}] at (3.5,-0.5) {};
\node[cir] (MMin2) [label=below:{\small $a_k\cdots a_l$}] at (5,-0.5) {};
\node[cir] (MMin3) [label=below:{\small $a_n$}] at (6.5,-0.5) {};
\node[poin] (VV)  at (5,0.6) {};
\draw[thick] (MMin1) -- (VV);
\draw[thick] (MMin2) -- (VV);
\draw[thick] (MMin3) -- (VV);
\draw[thick] (VV) -- (MMout);
\node[poin] at (4.25,-0.5) {{\small $\cdots$}};
\node[poin] at (5.75,-0.5) {{\small $\cdots$}};
\node[poin] at (2.5,0.5) {{$\sim$}};
\end{tikzpicture}
\end{flalign}
which we interpret as follows: Whenever $(i_k,i_{k+1},\dots,i_l) : (V_k, V_{k+1},\dots,V_l)\to V$ 
is a sub-family of embeddings that factorizes through a common interior region $V^\prime\subseteq V$,
then the formal product of the formal pushforward of observables is identified with the 
formal pushforward of the actual product of observables on $V^\prime$.
\end{rem}


\section{\label{sec:characterization}Characterization of boundary quantum field theories}
In the previous section we established a universal construction that allows
us to extend quantum field theories $\AAA\in\QFT(\interior{M})$ that are defined
only on the interior $\interior{M}$ of a spacetime $M$ with timelike boundary 
to the whole spacetime. The extension  $\ext\AAA\in\QFT(M)$ is characterized 
abstractly by the $\ext\dashv \res$ adjunction in
\eqref{eqn:EXTRES}. We now shall reverse the question and ask which quantum field theories
$\BBB\in \QFT(M)$ on $M$ admit a description in terms of (quotients of) our universal extensions.
\sk

Given any quantum field theory $\BBB\in\QFT(M)$ on the whole spacetime $M$,
we can use the right adjoint in \eqref{eqn:EXTRES} in order to restrict it to a theory
$\res\BBB \in\QFT(\interior{M})$ on the interior of $M$. Applying now the extension
functor, we obtain another quantum field theory 
$\ext\res\BBB \in \QFT(M)$ on the whole spacetime $M$, which we would like 
to compare to our original theory $\BBB\in\QFT(M)$. A natural
comparison map is given by the $\BBB$-component of the counit
$\epsilon : \ext\res \to\id_{\QFT(M)}$ of the adjunction \eqref{eqn:EXTRES}, 
i.e.\ the canonical $\QFT(M)$-morphism
\begin{subequations}\label{eqn:extrescounitB}
\begin{flalign}
\epsilon_{\BBB}^{} \,:\, \ext\res\BBB \longrightarrow \BBB\quad.
\end{flalign}
Using our model for the extension functor given in \eqref{eqn:extAAAV} (and the formulas
following this equation), the $\Reg_M[\Cauchy_M^{-1}]\ni V$-component of $\epsilon_{\BBB}^{}$
explicitly reads as
\begin{flalign}
(\epsilon_{\BBB}^{})_V^{} \,:\, \bigoplus_{\und{i}: \und{V}\to V}\!\!\BBB(\und{V}) \bigg/\!\!\!\sim  ~\longrightarrow~\BBB(V)~~,~~~~ [\und{i},\und{b}]\, \longmapsto \, \BBB(\und{i})\big(\und{b}\big)\quad.
\end{flalign}
\end{subequations}
In order to establish positive comparison results, we have to introduce the concept of ideals of quantum
field theories.
\begin{defi}\label{def:ideal}
An \textit{ideal} $\III\subseteq \BBB$ of a quantum field theory $\BBB\in\QFT(M)$
is a functor $\III : \Reg_{M}[\Cauchy_M^{-1}] \to \Vec$ to the 
category of complex vector spaces, which satisfies the following properties:
\begin{itemize}
\item[(i)] For all $V\in \Reg_{M}[\Cauchy_M^{-1}]$, $\III(V) \subseteq \BBB(V)$
is a two-sided $\ast$-ideal of the unital $\ast$-algebra $\BBB(V)$.

\item[(ii)] For all $\Reg_{M}[\Cauchy_M^{-1}]$-morphisms $i : V\to V^\prime$,
the linear map $\III (i):\III(V)\to \III(V^\prime)$ is the restriction of 
$\BBB(i) : \BBB(V)\to \BBB(V^\prime)$ to the two-sided $\ast$-ideals 
$\III(V^{(\prime)})\subseteq \BBB(V^{(\prime)})$.
\end{itemize}
\end{defi}
\begin{lem}
Let $\BBB\in\QFT(M)$ and $\III\subseteq \BBB$ any ideal. 
Let us define $\BBB/\III(V)  :=  \BBB(V)/\III(V)$ to be the quotient algebra, 
for all $V\in \Reg_{M}[\Cauchy_M^{-1}]$, and $\BBB/\III(i) : \BBB/\III(V)\to \BBB/\III(V^\prime)$ to be
the $\Alg$-morphism induced by $\BBB(i) : \BBB(V)\to \BBB(V^\prime)$, for all 
$\Reg_{M}[\Cauchy_M^{-1}]$-morphisms $i : V\to V^\prime$.
Then $\BBB/\III \in\QFT(M)$ is a quantum field theory on $M$ which we call
the quotient of $\BBB$ by $\III$.
\end{lem}
\begin{proof}
The requirements listed in Definition \ref{def:ideal} ensure that
$\BBB/\III : \Reg_{M}[\Cauchy_M^{-1}]\to \Alg$ is an $\Alg$-valued functor.
It satisfies the causality axiom of Definition \ref{def:QFT2} because this property is inherited from 
$\BBB\in\QFT(M)$ by taking quotients.
\end{proof}
\begin{lem}
Let $\kappa : \BBB\to\BBB^\prime$ be any $\QFT(M)$-morphism.
Define the vector space $\ker\kappa (V) := \ker\big(\kappa_V^{} : \BBB(V)\to \BBB^\prime(V)\big)
\subseteq \BBB(V)$, for all $V\in \Reg_{M}[\Cauchy_M^{-1}]$,
and $\ker\kappa(i) :\ker\kappa(V) \to \ker\kappa(V^\prime)$
to be the linear map induced by $\BBB(i):\BBB(V)\to\BBB(V^\prime)$,
for all $\Reg_{M}[\Cauchy_M^{-1}]$-morphisms $i : V\to V^\prime$.
Then $\ker\kappa : \Reg_{M}[\Cauchy_M^{-1}] \to \Vec$ is an ideal of $\BBB\in\QFT(M)$,
which we call the kernel of $\kappa$.
\end{lem}
\begin{proof}
The fact that $\ker\kappa$ defines a functor follows from naturality of $\kappa$.
Property (ii) in Definition \ref{def:ideal} holds true by construction.
Property (i) is a consequence of the fact that kernels of unital $\ast$-algebra morphisms
$\kappa_V^{} : \BBB(V)\to \BBB(V^\prime)$ are two-sided $\ast$-ideals.
\end{proof}

\begin{rem}
Using the concept of ideals, we may canonically factorize \eqref{eqn:extrescounitB} according to the diagram
\begin{flalign}\label{eqn:lambdafactorization}
\xymatrix{
\ar[dr]_-{\pi_\BBB^{}} \ext\res\BBB \ar[rr]^-{\epsilon_\BBB^{}} ~&&~ \BBB\\
&\ext\res\BBB\big/ \ker\epsilon_\BBB^{}\ar[ru]_-{\lambda_\BBB^{}}&
}
\end{flalign}
where both the projection $\pi_\BBB^{}$ and the inclusion $\lambda_\BBB^{}$ are $\QFT(M)$-morphisms.
\end{rem}

As a last ingredient for our comparison result, we have to introduce a suitable notion of additivity 
for quantum field theories on spacetimes with timelike boundary. We refer to 
\cite[Definition 2.3]{additivity} for a notion of additivity on globally hyperbolic spacetimes.
\begin{defi}\label{def:additive}
A quantum field theory $\BBB\in\QFT(M)$ on a spacetime $M$ with timelike boundary 
is called \textit{additive (from the interior) at the object $V\in\Reg_M[\Cauchy_M^{-1}]$} if
the algebra $\BBB(V)$ is generated by the images of the $\Alg$-morphisms
$\BBB(i_{\mathrm{int}}) : \BBB(V_{\mathrm{int}}) \to  \BBB(V)$, for all $\Reg_{M}[\Cauchy_M^{-1}]$-morphism
$i_{\mathrm{int}} : V_{\mathrm{int}}\to V$ whose source $V_{\mathrm{int}}\in\Reg_{\interior{M}}[\Cauchy_{\interior{M}}^{-1}]$ 
is in the interior $\interior{M}$ of $M$. We call $\BBB\in\QFT(M)$ \textit{additive (from the interior)} if it is additive
at every object $V\in\Reg_M[\Cauchy_M^{-1}]$. The full subcategory of additive quantum field theories
on $M$ is denoted by $\QFT^{\mathrm{add}}(M)\subseteq \QFT(M)$.
\end{defi}

We can now prove our first characterization theorem for boundary quantum field theories.
\begin{theo}\label{theo:quotientVSadditivity}
Let $\BBB\in\QFT(M)$ be any quantum field theory on a (not necessarily globally hyperbolic) 
spacetime $M$ with timelike boundary and let $V\in \Reg_{M}[\Cauchy_M^{-1}]$. 
Then the following are equivalent:
\begin{itemize}
\item[(1)] The $V$-component
\begin{flalign}
(\lambda_\BBB^{})_V^{} \,:\, \ext\res\BBB(V)\big/ \ker\epsilon_\BBB^{}(V)~\longrightarrow~\BBB(V)
\end{flalign}
of the canonical inclusion in \eqref{eqn:lambdafactorization} is an $\Alg$-isomorphism.

\item[(2)] $\BBB$ is additive at the object  $V\in \Reg_{M}[\Cauchy_M^{-1}]$.
\end{itemize}
\end{theo}
\begin{proof}
Let $i_{\mathrm{int}} : V_{\mathrm{int}} \to V$ be any  $\Reg_{M}[\Cauchy_M^{-1}]$-morphism
whose source $V_{\mathrm{int}}\in\Reg_{\interior{M}}[\Cauchy_{\interior{M}}^{-1}]$ 
is in the interior $\interior{M}$ of $M$. Using our model for the extension 
functor given in \eqref{eqn:extAAAV} (and the formulas following this equation),
we obtain an $\Alg$-morphism
\begin{flalign}\label{eqn:iintinclusion}
[i_{\mathrm{int}} , - ] \,:\, \BBB(V_{\mathrm{int}}) \longrightarrow \ext\res\BBB(V)~,~~
b\longmapsto [i_{\mathrm{int}},b]\quad.
\end{flalign}
Composing this morphism with the $V$-component of $\epsilon_\BBB^{}$ given in \eqref{eqn:extrescounitB}, 
we obtain a commutative diagram
\begin{flalign}\label{eqn:comporule}
\xymatrix{
\ext\res\BBB(V) \ar[rr]^-{(\epsilon_\BBB^{})_V^{}} && \BBB(V)\\
&\ar[ul]^-{ [i_{\mathrm{int}} , - ]~~~}\BBB(V_{\mathrm{int}}) \ar[ru]_-{~~~\BBB(i_{\mathrm{int}})}&
}
\end{flalign}
for all $i_{\mathrm{int}} : V_{\mathrm{int}} \to V$ with $V_{\mathrm{int}}$ in the interior.
\sk

Next we observe that the images of the $\Alg$-morphisms \eqref{eqn:iintinclusion}, for all  
$i_{\mathrm{int}} : V_{\mathrm{int}} \to V$ with $V_{\mathrm{int}}$ in the interior, 
generate $\ext\res\BBB(V)$. Combining this property with \eqref{eqn:comporule}, we conclude that
$(\epsilon_\BBB^{})_V^{}$ is a surjective map if and only if  $\BBB$ is additive at $V$.
Hence, $(\lambda_\BBB^{})_V^{} $ given by \eqref{eqn:lambdafactorization}, which is injective by construction, 
is an $\Alg$-isomorphism if and only if  $\BBB$ is additive at $V$.
\end{proof}
\begin{cor}\label{cor:quotientVSadditivity}
$\lambda_\BBB^{} : \ext\res\BBB\big/\ker\epsilon_\BBB^{} \to \BBB$ given by \eqref{eqn:lambdafactorization}
is a $\QFT(M)$-isomorphism if and only if $\BBB\in \QFT^{\mathrm{add}}(M)\subseteq\QFT(M)$ 
is additive in the sense of Definition \ref{def:additive}.
\end{cor}

We shall now refine this characterization theorem by showing that
$\QFT^{\mathrm{add}}(M)$ is equivalent, as a category, to a category
describing quantum field theories on the interior of $M$ together with 
suitable ideals of their universal extensions. The precise definitions are as follows.
\begin{defi}\label{defi:trivinteriorideal}
Let $\BBB\in\QFT(M)$. An ideal $\III\subseteq \BBB$ is called \textit{trivial on the interior}
if its restriction to $\Reg_{\interior{M}}[\Cauchy_{\interior{M}}^{-1}]$ is the functor assigning
zero vector spaces, i.e.\ $\III(V_{\mathrm{int}}) = 0$ for all $V_{\mathrm{int}}\in 
\Reg_{\interior{M}}[\Cauchy_{\interior{M}}^{-1}]$ in the interior $\interior{M}$ of $M$.
\end{defi}
\begin{defi}
Let $M$ be a spacetime with timelike boundary. 
We define the category $\IQFT(M)$ as follows:
\begin{itemize}
\item Objects are pairs $(\AAA,\III)$ consisting of a quantum field theory 
$\AAA\in\QFT(\interior{M})$ on the interior $\interior{M}$ of $M$ and an
ideal $\III\subseteq \ext\AAA$ of its universal extension $\ext\AAA\in\QFT(M)$ that
is trivial on the interior.

\item Morphisms $\kappa : (\AAA,\III)\to (\AAA^\prime,\III^\prime)$
are $\QFT(M)$-morphisms $\kappa : \ext\AAA\to \ext\AAA^{\prime}$
between the universal extensions that preserve the ideals,
i.e.\ $\kappa$ restricts to a natural transformation 
from $\III\subseteq \ext\AAA$ to $\III^\prime\subseteq \ext\AAA^\prime$.
\end{itemize}
\end{defi}

\noindent There exists an obvious functor
\begin{subequations}\label{eqn:Qfunctor}
\begin{flalign}
Q \,:\, \IQFT(M)\longrightarrow \QFT^{\mathrm{add}}(M)\quad,
\end{flalign}
which assigns to $(\AAA,\III)\in\IQFT(M)$ the quotient
\begin{flalign}
Q(\AAA,\III)\,:=\, \ext\AAA\big/\III \in\QFT^{\mathrm{add}}(M)\quad.
\end{flalign}
\end{subequations}
Notice that additivity of $\ext\AAA\big/\III$ follows from that of the universal extension
$\ext\AAA$ (cf.\ the arguments in the proof of Theorem \ref{theo:quotientVSadditivity})
and the fact that quotients preserve the additivity property. There exists also a functor 
\begin{subequations}\label{eqn:Sfunctor}
\begin{flalign}
S\,:\,  \QFT^{\mathrm{add}}(M)\longrightarrow \IQFT(M)\quad,
\end{flalign}
which `extracts' from a quantum field theory on $M$ the relevant ideal. Explicitly, it
assigns to $\BBB\in\QFT^{\mathrm{add}}(M)$ the pair
\begin{flalign}
S\BBB \,:=\, \big(\res\BBB , \ker\epsilon_\BBB^{}\big)\quad.
\end{flalign}
\end{subequations}
Notice that the ideal $\ker\epsilon_\BBB^{} \subseteq \ext\res\BBB$ is trivial on the interior: Applying
the restriction functor \eqref{eqn:EXTRES} to $\epsilon_\BBB^{}$ we obtain
a $\QFT(\interior{M})$-morphism
\begin{flalign}\label{eqn:resepsilonB}
\res\epsilon_{\BBB}^{} \,:\,\res \ext\res\BBB\longrightarrow \res\BBB\quad.
\end{flalign}
Proposition \ref{propo:EXTRESunit} together with the 
triangle identities for the adjunction $\ext\dashv\res$ in \eqref{eqn:EXTRES}
then imply that \eqref{eqn:resepsilonB} is an isomorphism with inverse
given by $\eta_{\res\BBB}^{} :  \res\BBB\to \res \ext\res\BBB$.
In particular, $\res\epsilon_{\BBB}^{}$ has a trivial kernel and hence
$\ker\epsilon_\BBB^{} \subseteq \ext\res\BBB$ is trivial on the interior (cf.\ Definition \ref{defi:trivinteriorideal}).
Our refined characterization theorem for boundary quantum field theories
is as follows.
\begin{theo}\label{theo:IQFTisQFT}
The functors $Q$ and $S$ defined in \eqref{eqn:Qfunctor} and \eqref{eqn:Sfunctor}
exhibit an equivalence of categories
\begin{flalign}
 \QFT^{\mathrm{add}}(M)\,\cong\,  \IQFT(M)\quad.
\end{flalign}
\end{theo}
\begin{proof}
We first consider the composition of functors $Q\,S :\QFT^{\mathrm{add}}(M)\to \QFT^{\mathrm{add}}(M)$.
To $\BBB\in \QFT^{\mathrm{add}}(M)$, it assigns
\begin{flalign}
QS\BBB = \ext\res\BBB\big/\ker\epsilon_\BBB^{}\quad.
\end{flalign}
The $\QFT(M)$-morphisms $\lambda_\BBB^{} : \ext\res\BBB/\ker\epsilon_\BBB^{}\to \BBB$
given by \eqref{eqn:lambdafactorization} define a natural transformation
$\lambda : Q\,S \to \id_{\QFT^{\mathrm{add}}(M)}$, which
is a natural isomorphism  due to Corollary \ref{cor:quotientVSadditivity}.
\sk

Let us now consider the composition of functors $S\,Q : \IQFT(M)\to \IQFT(M)$. To 
$(\AAA,\III)\in \IQFT(M)$, it assigns
\begin{flalign}
SQ(\AAA,\III) = \Big(\res \big(\ext\AAA\big/ \III\big),\ker\epsilon_{\ext\AAA/ \III}^{}\Big)
= \Big(\res \ext\AAA,\ker\epsilon_{\ext\AAA/ \III}^{}\Big)\quad,
\end{flalign}
where we also used that $\res \big(\ext\AAA\big/ \III\big) = \res\ext\AAA$
because $\III$ is by hypothesis trivial on the interior. 
Using further the $\QFT(\interior{M})$-isomorphism $\eta_{\AAA} : \AAA\to \res\ext\AAA$
from Proposition \ref{propo:EXTRESunit}, we define a $\QFT(M)$-morphism $q_{(\AAA,\III)}$
via the diagram
\begin{flalign}\label{eqn:projdiagramtmp}
\xymatrix{
\ar[d]^-{\cong}_-{\ext\eta_\AAA^{}}\ext\AAA \ar[rr]^-{q_{(\AAA,\III)}} && \ext\AAA\big/\III\\
\ext\res \ext\AAA \ar@{=}[rr]&& \ext \res \big(\ext\AAA\big/ \III\big)\ar[u]_-{\epsilon_{\ext\AAA/\III}^{}}
}
\end{flalign}
Using the explicit expression for $\epsilon_{\ext\AAA/\III}^{}$
given in \eqref{eqn:extrescounitB} and the 
explicit expression for $\eta_\AAA^{}$ given by
\begin{flalign}
(\eta_{\AAA}^{})_{V_{\mathrm{int}}}^{} \,:\,\AAA(V_{\mathrm{int}}) \longrightarrow \res\ext\AAA(V_{\mathrm{int}}) = 
\bigoplus_{\und{i}: \und{V}\to V_{\mathrm{int}}}\!\!\AAA(\und{V}) \bigg/\!\!\!\sim 
~~,~~~~a\longmapsto [\id_{V_{\mathrm{int}}},a]\quad,
\end{flalign}
for all $V_{\mathrm{int}}\in \Reg_{\interior{M}}[\Cauchy_{\interior{M}}^{-1}]$,
one computes from the diagram \eqref{eqn:projdiagramtmp} that $q_{(\AAA,\III)}$ is the canonical
projection $\pi : \ext \AAA\to \ext\AAA/\III$. Hence, the $\QFT(M)$-isomorphisms 
$\ext\eta_{\AAA}^{} : \ext\AAA \to\ext\res\ext\AAA$ induce $\IQFT(M)$-isomorphisms
\begin{flalign}
\ext\eta_{\AAA}^{} : (\AAA,\III)\longrightarrow SQ(\AAA,\III)\quad,
\end{flalign}
which are natural in $(\AAA,\III)$, i.e.\ they define a natural isomorphism $\ext\eta : \id_{\IQFT(M)}\to S\,Q$.
\end{proof}

\begin{rem}\label{rem:bdyphysics}
The physical interpretation of this result is as follows:
Every additive quantum field theory $\BBB\in\QFT^{\mathrm{add}}(M)$ on a (not necessarily globally hyperbolic) 
spacetime $M$ with timelike boundary admits an equivalent description in terms of a pair $(\AAA,\III)\in\IQFT(M)$. 
Notice that the roles of $\AAA$ and $\III$ are completely different:
On the one hand, $\AAA\in\QFT(\interior{M})$ is a quantum field theory
on the interior $\interior{M}$ of $M$ and as such it is independent of
the detailed aspects of the boundary. On the other hand, $\III\subseteq \ext\AAA$
is an ideal of the universal extension of $\AAA$ that is trivial on the interior,
i.e.\ it only captures the physics that happens directly at the boundary. 
Examples of such ideals $\III$ arise 
by imposing specific boundary conditions 
on the universal extension $\ext\AAA\in \QFT(M)$, i.e.\ the quotient
$\ext\AAA/\III$ describes a quantum field theory on $M$ that satisfies specific 
boundary conditions encoded in $\III$. We shall illustrate this assertion in 
Section \ref{sec:KG} below using the explicit example given by the free Klein-Gordon field.
\sk

Let us also note that there is a reason why our universal extension 
captures only the class of additive quantum field theories on $M$. Recall
that $\ext\AAA\in\QFT(M)$ takes as an input a quantum field theory 
$\AAA\in\QFT(\interior{M})$ on the interior $\interior{M}$ of $M$. As a consequence, 
the extension $\ext\AAA$ can only have knowledge  of 
the `degrees of freedom' that are generated in some way out of the interior regions. Additive theories in 
the sense of Definition \ref{def:additive} are precisely the theories 
whose `degrees of freedom' are generated out of those localized in the interior regions.
\end{rem}

\section{\label{sec:KG}Example: Free Klein-Gordon theory}
In order to illustrate and make more explicit our abstract constructions 
developed in the previous sections, we shall consider the simple example
given by the free Klein-Gordon field. From now on $M$ 
will be a \textit{globally hyperbolic} spacetime with timelike boundary, 
see Definition \ref{def:globhyp}. This assumption implies that 
all interior regions $\Reg_{\interior{M}}$ are
globally hyperbolic spacetimes with empty boundary, see Proposition \ref{propo:globhyp}. 
This allows us to use the standard techniques of \cite[Section 3]{BGP} 
on such regions. 

\subsection*{Definition on $\Reg_{\interior{M}}[\Cauchy_{\interior{M}}^{-1}]$:}
Let $M$ be a globally hyperbolic spacetime with timelike boundary, 
see Definition \ref{def:globhyp}. The free Klein-Gordon
theory on $\Reg_{\interior{M}}[\Cauchy_{\interior{M}}^{-1}]$ is given by the following 
standard construction, see e.g.\ \cite{BD,BDH} for expository reviews.
On the interior $\interior{M}$, we consider the Klein-Gordon operator
\begin{flalign}\label{eqn:PintM}
P \,:= \, \square + m^2 ~:~C^\infty(\interior{M})\longrightarrow C^\infty(\interior{M})\quad,
\end{flalign} 
where $\square$ is the d'Alembert operator and $m\geq 0$ is a mass parameter.
When restricting $P$ to regions $V\in \Reg_{\interior{M}}[\Cauchy_{\interior{M}}^{-1}]$, we shall write
\begin{flalign}\label{eqn:PV}
P_{V}^{} ~:~C^\infty(V)\longrightarrow C^\infty(V)\quad.
\end{flalign}
It follows from \cite{BGP} that there exists a unique retarded/advanced Green's operator 
\begin{flalign}
G_{V}^\pm ~:~  C_\cc^\infty(V)\longrightarrow C^\infty(V)
\end{flalign}
for $P_V$ because every $V\in \Reg_{\interior{M}}[\Cauchy_{\interior{M}}^{-1}]$ is a 
globally hyperbolic spacetime with empty boundary, cf.\ Proposition \ref{propo:globhyp}.
\sk

The Klein-Gordon theory $\KKK \in \QFT(\interior{M})$
is the functor $\KKK : \Reg_{\interior{M}}[\Cauchy_{\interior{M}}^{-1}] \to\Alg$ 
given by the following assignment: To any $V\in \Reg_{\interior{M}}[\Cauchy_{\interior{M}}^{-1}]$ 
it assigns the associative and unital $\ast$-algebra $\KKK(V)$ that is freely generated by $\Phi_V(f)$, 
for all $f\in C^\infty_\cc(V)$, modulo the two-sided $\ast$-ideal generated by the following relations:
\begin{itemize}
\item \textit{Linearity:} $\Phi_V(\alpha\,f + \beta \,g) = \alpha\,\Phi_V(f) + \beta \,\Phi_V(g)$,
for all $\alpha,\beta\in\bbR$ and $f,g\in C^\infty_\cc(V)$;
\item \textit{Hermiticity:} $\Phi_V(f)^\ast = \Phi_V(f)$, for all $f\in C^\infty_\cc(V)$;
\item \textit{Equation of motion:} $\Phi_V(P_V f) = 0$, for all $f\in C^\infty_\cc(V)$;
\item \textit{Canonical commutation relations (CCR):} $\Phi_V(f) \, \Phi_V(g) - \Phi_V(g) \,\Phi_V(f) 
= \mathrm{i} \,\tau_V(f,g)~\oone$, for all $f,g\in C^\infty_\cc(V)$, where
\begin{flalign}
\tau_V(f,g) \,:=\, \int_V f \,G_V (g)~\vol_V
\end{flalign}
with $G_V := G_V^+ - G_V^-$ the causal propagator and $\vol_V$ the canonical volume form on $V$. 
(Note that $\tau_V$ is antisymmetric, see e.g.\ \cite[Lemma 4.3.5]{BGP}.) 
\end{itemize}
To a morphism $i : V\to V^\prime$ in $\Reg_{\interior{M}}[\Cauchy_{\interior{M}}^{-1}]$, 
the functor  $\KKK :\Reg_{\interior{M}}[\Cauchy_{\interior{M}}^{-1}] \to\Alg$
assigns the algebra map that is specified on the generators by pushforward along $i$ (which we shall suppress)
\begin{flalign}
\KKK(i) \,:\, \KKK(V)\longrightarrow \KKK(V^\prime)~,~~\Phi_V(f) \longmapsto \Phi_{V^\prime}(f)\quad.
\end{flalign}
The naturality of $\tau$ 
(i.e.\ naturality of the causal propagator, cf.\ e.g.\ \cite[Section 4.3]{BGP}) 
entails that the assignment $\KKK$ defines a quantum field theory 
in the sense of Definition \ref{def:QFT2}.

\subsection*{Universal extension:}
Using the techniques developed in Section \ref{sec:bdyext}, we may now extend 
the Klein-Gordon theory $\KKK\in\QFT(\interior{M})$ from the interior $\interior{M}$
to the whole spacetime $M$. In particular, using \eqref{eqn:extAAAV} (and the 
formulas following this equation), one could directly compute the universal extension 
$\ext\KKK \in\QFT(M)$. The resulting expressions, however, can be considerably simplified. We therefore prefer to provide
a more convenient model for the universal extension $\ext\KKK \in\QFT(M)$
by adopting the following strategy: We first make an `educated guess' for a
theory ${\KKK}^{\ext} \in \QFT(M)$ which we expect to be the universal extension
of $\KKK\in \QFT(\interior{M})$. (This was inspired by partially simplifying the 
direct computation of the universal extension.) After this we shall prove that
${\KKK}^{\ext} \in \QFT(M)$ satisfies the universal property that characterizes
$\ext\KKK \in\QFT(M)$. Hence, there exists a (unique) isomorphism
$\ext\KKK \cong {\KKK}^{\ext}$ in $\QFT(M)$, which means that our
${\KKK}^{\ext} \in\QFT(M)$ is a model for the universal extension $\ext\KKK$.
\sk

Let us define the functor ${\KKK}^{\ext} : \Reg_{M}[\Cauchy_M^{-1}]\to\Alg$ by the following
assignment: To any region $V\in \Reg_M[\Cauchy_M^{-1}]$, which may intersect the boundary,
we assign the associative and unital $\ast$-algebra ${\KKK}^{\ext}(V)$ that is freely generated
by $\Phi_V(f)$, for all $f\in C^\infty_\cc(\interior{V})$ in the interior $\interior{V}$ of $V$,
modulo the two-sided ideal generated by the following relations:
\begin{itemize}
\item \textit{Linearity:} $\Phi_V(\alpha\,f + \beta \,g) = \alpha\,\Phi_V(f) + \beta \,\Phi_V(g)$,
for all $\alpha,\beta\in\bbR$ and $f,g\in C^\infty_\cc(\interior{V})$;
\item \textit{Hermiticity:} $\Phi_V(f)^\ast = \Phi_V(f)$, for all $f\in C^\infty_\cc(\interior{V})$;
\item \textit{Equation of motion:} $\Phi_V(P_{\interior{V}} f) = 0$, for all $f\in C^\infty_\cc(\interior{V})$;
\item \textit{Partially-defined CCR:} $\Phi_V(f) \, \Phi_V(g) - \Phi_V(g) \,\Phi_V(f) = \mathrm{i} \,
\tau_{V_{\mathrm{int}} }(f,g)~\oone$, for all interior regions 
$V_{\mathrm{int}} \in \Reg_{\interior{M}}[\Cauchy_{\interior{M}}^{-1}]$
with $V_{\mathrm{int}}\subseteq  \interior{V} $ and  $f,g\in C^\infty_\cc(\interior{V})$ 
with $\supp(f)\cup\supp(g)\subseteq V_{\mathrm{int}}$.
\end{itemize}
\begin{rem}
We note that our partially-defined CCR are consistent in the following sense:
Consider $V_{\mathrm{int}}, V_{\mathrm{int}}^\prime \in \Reg_{\interior{M}}[\Cauchy_{\interior{M}}^{-1}]$
with $V_{\mathrm{int}}^{(\prime)}\subseteq  \interior{V} $  and $f,g\in C^\infty_\cc(\interior{V})$ 
with the property that $\supp(f) \cup \supp(g) \subseteq V_{\mathrm{int}}^{(\prime)}$.
Using the partially-defined CCR for both $V_{\mathrm{int}}$ and $V_{\mathrm{int}}^\prime$, we obtain the equality
$\mathrm{i} \,\tau_{V_{\mathrm{int}}}(f,g)~\oone = \mathrm{i} \,\tau_{V_{\mathrm{int}}^\prime}(f,g)~\oone$ 
in ${\KKK}^{\ext}(V)$. To ensure that ${\KKK}^{\ext}(V)$ is not the zero-algebra, we have to show
that $\tau_{V_{\mathrm{int}}}(f,g)=\tau_{V_{\mathrm{int}}^\prime}(f,g)$. This holds true due to the following
argument: Consider the subset $V_{\mathrm{int}}\cap V_{\mathrm{int}}^{\prime} \subseteq \interior M$. 
This is open, causally convex 
and by Proposition \ref{propo:Cauchy1}~(c) also stable under Cauchy development, 
hence $V_{\mathrm{int}}\cap V_{\mathrm{int}}^{\prime} \in 
\Reg_{\interior{M}}[\Cauchy_{\interior{M}}^{-1}]$. Furthermore, the inclusions 
$V_{\mathrm{int}}\cap V_{\mathrm{int}}^{\prime}\to V_{\mathrm{int}}^{(\prime)}$
are morphisms in $\Reg_{\interior{M}}[\Cauchy_{\interior{M}}^{-1}] $. It follows by construction that 
$\supp(f) \cup \supp(g) \subseteq V_{\mathrm{int}}\cap V_{\mathrm{int}}^\prime$ and hence 
due to naturality of the $\tau$'s we obtain
\begin{flalign}
\tau_{V_{\mathrm{int}}}(f,g)= \tau_{V_{\mathrm{int}}\cap V_{\mathrm{int}}^\prime}(f,g) =
 \tau_{V_{\mathrm{int}}^\prime}(f,g)\quad.
\end{flalign}
Hence, for any fixed pair $f,g\in C^\infty_\cc(\interior{V})$, the partially-defined CCR are independent
of the choice of $V_{\mathrm{int}}$ (if one exists).
\end{rem}

To a morphism $i : V\to V^\prime$ in $\Reg_{M}[\Cauchy_{M}^{-1}]$, the functor
${\KKK}^{\ext} : \Reg_{M}[\Cauchy_M^{-1}]\to\Alg$ assigns the algebra map that is specified
on the generators by the pushforward  along $i$ (which we shall suppress)
\begin{flalign}\label{eqn:KKKextmaps}
{\KKK}^{\ext}(i) \, :\,  {\KKK}^{\ext}(V) \longrightarrow {\KKK}^{\ext}(V^\prime)~,~~
\Phi_V(f) \longmapsto \Phi_{V^\prime}(f)\quad.
\end{flalign}
Compatibility of the map \eqref{eqn:KKKextmaps} with the relations in ${\KKK}^{\ext}$ 
is a straightforward check.
\sk

Recalling the embedding functor $J : \Reg_{\interior{M}}[\Cauchy_{\interior{M}}^{-1}] \to
\Reg_{M}[\Cauchy_M^{-1}]$ given in \eqref{eqn:J}, we observe that
the diagram of functors
\begin{flalign}\label{eqn:LanJKKK}
\xymatrix{
\ar[dr]_-{J}\Reg_{\interior{M}}[\Cauchy_{\interior{M}}^{-1}]  \ar[rr]^-{\KKK} &\ar@{=>}[d]^-{\gamma}& \Alg\\
&\Reg_{M}[\Cauchy_M^{-1}] \ar[ur]_-{{\KKK}^{\ext}} &
}
\end{flalign}
commutes via the natural transformation $\gamma : \KKK \to {\KKK}^{\ext}\, J$ with components
specified on the generators by the identity maps
\begin{flalign}\label{eqn:LanJKKKunit}
\gamma_{V_{\mathrm{int}}}^{}\,:\,  \KKK(V_{\mathrm{int}}) \longrightarrow {\KKK}^{\ext}(V_{\mathrm{int}})~,~~
\Phi_{V_{\mathrm{int}}}(f) \longmapsto \Phi_{V_{\mathrm{int}}}(f)\quad,
\end{flalign}
for all $V_{\mathrm{int}} \in \Reg_{\interior{M}}[\Cauchy_{\interior{M}}^{-1}]$. Notice that
$\gamma$ is a natural isomorphism because $\interior{V_{\mathrm{int}}} = V_{\mathrm{int}}$ 
and the partially-defined CCR on any interior region $V_{\mathrm{int}}$ 
coincides with the CCR.
\begin{theo}
\eqref{eqn:LanJKKK} is a left Kan extension of $\KKK: \Reg_{\interior{M}}[\Cauchy_{\interior{M}}^{-1}] \to
\Alg $ along $J : \Reg_{\interior{M}}[\Cauchy_{\interior{M}}^{-1}] \to \Reg_{M}[\Cauchy_M^{-1}]$.
As a consequence of uniqueness (up to unique natural isomorphism) of left Kan extensions
and Proposition \ref{propo:LanJ}, it follows that ${\KKK}^{\ext} \cong \ext \KKK$, i.e.\
${\KKK}^{\ext} \in \QFT(M)$ is a model for our universal extension $\ext\KKK\in\QFT(M)$
of the Klein-Gordon theory $\KKK\in \QFT(\interior{M})$. 
\end{theo}
\begin{proof}
We have to prove that \eqref{eqn:LanJKKK} satisfies the universal property of left Kan extensions:
Given any functor $\BBB : \Reg_{M}[\Cauchy_M^{-1}] \to \Alg$ and natural transformation $\rho: 
\KKK \to \BBB\,J$, we have to show that there exists a 
unique natural transformation $\zeta : {\KKK}^{\ext}\to \BBB$ such that the diagram
\begin{flalign}
\xymatrix{
\ar[dr]_-{\gamma} \KKK \ar[rr]^-{\rho} && \BBB\,J\\
&{\KKK}^{\ext}\,J \ar[ru]_-{\zeta J}&
}
\end{flalign}
commutes. Because $\gamma$ is a natural isomorphism, it immediately follows that
$\zeta J$ is uniquely fixed by this diagram. Concretely, this means that 
the components $\zeta_{V_{\mathrm{int}}}$ corresponding to interior regions
$V_{\mathrm{int}}\in\Reg_{\interior{M}}[\Cauchy_{\interior{M}}^{-1}]$ are uniquely
fixed by
\begin{flalign}\label{eqn:zetaVint}
\zeta_{V_{\mathrm{int}}}^{}:= \rho_{V_{\mathrm{int}}}^{}\, \gamma_{V_{\mathrm{int}}}^{-1}  \,:\, {\KKK}^{\ext}(V_{\mathrm{int}}) \longrightarrow \BBB(V_{\mathrm{int}})\quad.
\end{flalign}
It remains to determine the components 
\begin{flalign}
\zeta_{V} \,:\, {\KKK}^{\ext}(V) \longrightarrow \BBB(V)
\end{flalign}
for generic regions $V\in \Reg_{M}[\Cauchy_M^{-1}]$. Consider any generator
$\Phi_V(f)$ of  ${\KKK}^{\ext}(V)$, where $f\in C^\infty_\cc(\interior{V})$,
and choose a finite cover $\{V_\alpha\subseteq \interior{V}\}$ of $\supp(f)$ 
by interior regions $V_\alpha \in\Reg_{\interior{M}}[\Cauchy_{\interior{M}}^{-1}]$, together 
with a partition of unity $\{\chi_\alpha\}$ subordinate to this cover. 
(The existence of such a cover is guaranteed by the assumption 
that $M$ is a globally hyperbolic spacetime with timelike boundary, 
see Proposition \ref{propo:globhyp}.) We define
\begin{flalign}\label{eqn:zetaV}
\zeta_{V} \big(\Phi_V(f)\big) \,:=\, \sum_\alpha \BBB(i_\alpha)\big(\zeta_{V_\alpha}\big( \Phi_{V_\alpha}(\chi_\alpha f)\big)\big)\quad,
\end{flalign}
where $i_\alpha :V_\alpha \to V$ is the inclusion. Our definition \eqref{eqn:zetaV} is independent of the
choice of cover and partition of unity: For any other 
$\{V^\prime_\beta\subseteq \interior{V}\}$ and $\{\chi^\prime_\beta\}$, we obtain
\begin{flalign}
\nn \sum_\beta \BBB(i_\beta)\big(\zeta_{V_\beta^\prime}\big( \Phi_{V_\beta^\prime}(\chi_\beta^\prime f)\big)\big)
&= \sum_{\alpha, \beta} \BBB(i_{\alpha\beta})\big(\zeta_{V_\alpha\cap V_\beta^\prime}\big( \Phi_{V_\alpha\cap V_\beta^\prime}(\chi_\alpha \chi_\beta^\prime f)\big)\big)\\
&=\sum_\alpha \BBB(i_\alpha)\big(\zeta_{V_\alpha}\big( \Phi_{V_\alpha}(\chi_\alpha f)\big)\big)\quad,
\end{flalign}
where $i_\beta : V_\beta^\prime \to V$ and $i_{\alpha\beta} : V_\alpha\cap V_\beta^\prime\to V$
are the inclusions. In particular, this implies that \eqref{eqn:zetaV} coincides with \eqref{eqn:zetaVint}
on interior regions $V=V_{\mathrm{int}}$. (Hint: Choose the cover given by the single region 
$V_{\mathrm{int}}$ together with its partition of unity.)
\sk

We have to check that \eqref{eqn:zetaV} preserves the relations in ${\KKK}^{\ext}(V)$. 
Preservation of linearity and Hermiticity is obvious.
The equation of motion relations are preserved because
\begin{flalign}
\nn \zeta_{V} \big(\Phi_V(P_{\interior{V}}^{} f)\big) &= 
\sum_\alpha \BBB(i_\alpha)\big(\zeta_{V_\alpha}\big( \Phi_{V_\alpha}(\chi_\alpha P_{V_\alpha}^{} f)\big)\big)\\
\nn &=\sum_{\alpha, \beta} \BBB(i_{\alpha\beta})\big(\zeta_{V_\alpha\cap V_\beta}\big( \Phi_{V_\alpha\cap V_\beta}(\chi_\alpha P_{V_\alpha\cap V_\beta}^{} \chi_\beta f)\big)\big)\\
&=\sum_{\beta} \BBB(i_{\beta})\big(\zeta_{V_\beta}\big( \Phi_{V_\beta}(P_{V_\beta}^{} \chi_\beta f)\big)\big) =0\quad.
\end{flalign}
Regarding the partially-defined CCR,
let $V_{\mathrm{int}} \in \Reg_{\interior{M}}[\Cauchy_{\interior{M}}^{-1}]$
with $V_{\mathrm{int}}\subseteq \interior{V}$ and $f,g\in C^\infty_\cc(\interior{V})$ 
with $\supp(f)\cup \supp(g) \subseteq V_{\mathrm{int}}$. We may choose the cover given by the 
single region $i: V_{\mathrm{int}}\to V$ together with its partition of unity.  We obtain for the commutator
\begin{flalign}
\Big[\zeta_{V} \big(\Phi_V(f)\big), \zeta_{V} \big(\Phi_V(g)\big)\Big]
= \Big[\BBB(i)\big(\zeta_{V_{\mathrm{int}}}\big( \Phi_{V_{\mathrm{int}}}(f)\big)\big) , 
\BBB(i)\big(\zeta_{V_{\mathrm{int}}}\big( \Phi_{V_{\mathrm{int}}}(g)\big)\big) \Big]
= \mathrm{i}\, \tau_{V_{\mathrm{int}}}(f,g)\,\oone\quad,
\end{flalign}
which implies that the partially-defined CCR are preserved.
\sk

Naturality of the components \eqref{eqn:zetaV} is easily verified. Uniqueness of the resulting 
natural transformation $\zeta : {\KKK}^{\ext}\to \BBB$ is a consequence of uniqueness of $\zeta J$
and of the fact that the $\Alg$-morphisms ${\KKK}^{\ext}(i_{\mathrm{int}}) : {\KKK}^{\ext}(V_{\mathrm{int}}) \cong
\KKK(V_{\mathrm{int}}) \to {\KKK}^{\ext}(V)$, for all interior regions $i_{\mathrm{int}} : V_{\mathrm{int}}\to V$,
generate ${\KKK}^{\ext}(V)$, for all $V\in \Reg_{M}[\Cauchy_M^{-1}]$. This completes the proof.
\end{proof}

\subsection*{Ideals from Green's operator extensions:}
The Klein-Gordon theory $\KKK\in\QFT(\interior{M})$ on the interior 
$\interior{M}$ of the globally hyperbolic spacetime $M$ with timelike boundary 
and its universal extension ${\KKK}^{\ext}\in\QFT(M)$ depend on the \textit{local}
retarded and advanced Green's operators $G_{V_{\mathrm{int}}}^\pm : C^\infty_\cc(V_{\mathrm{int}})
\to C^\infty(V_{\mathrm{int}})$ on all interior regions $V_{\mathrm{int}}^{} 
\in \Reg_{\interior{M}}[\Cauchy_{\interior{M}}^{-1}]$ in $M$. 
For constructing concrete examples of quantum field theories on globally hyperbolic spacetimes 
with timelike boundary as in \cite{Casimir}, one typically imposes suitable boundary conditions
for the field equation in order to obtain also \textit{global} retarded
and advanced Green's operators on $M$. Inspired by such examples,
we shall now show that any choice of an \textit{adjoint-related pair} $(G^+,G^-)$ 
consisting of a retarded and an advanced Green's operator for $P$ 
on $M$ (see Definition \ref{defi:Green} below) defines an ideal 
$\III_{G^\pm} \subseteq {\KKK}^{\ext}\in\QFT(M)$ that is trivial on the 
interior (cf.\ Definition \ref{defi:trivinteriorideal}). The corresponding quotient
${\KKK}^{\ext}\big/ \III_{G^\pm}\in\QFT(M)$ then may be interpreted as the Klein-Gordon theory on $M$,
subject to a specific choice of boundary conditions that is encoded in $G^\pm$.
\begin{defi}\label{defi:Green}
A \textit{retarded/advanced Green's operator} for the Klein-Gordon operator 
$P$ on $M$ is a linear
map $G^\pm : C^\infty_\cc(\interior{M}) \to C^\infty(\interior{M})$  which satisfies the
following properties, for all $f\in C^\infty_\cc(\interior{M})$:
\begin{itemize}
\item[(i)] $P\, G^\pm (f) = f$, 
\item[(ii)] $ G^\pm  (P f) = f$, and
\item[(iii)] $\supp (G^\pm (f) ) \subseteq J^\pm_M(\supp(f))$.
\end{itemize}
A pair $(G^+,G^-)$ consisting of a retarded and an advanced Green's operator 
for $P$ on $M$ is called \textit{adjoint-related} if $G^+$ is the formal adjoint of $G^-$, i.e.\
\begin{flalign}
\int_M G^+ (f) \,g~\vol_M = \int_M f \,G^- (g)~\vol_M \quad,
\end{flalign}
for all $f,g \in C^\infty_\cc(\interior{M})$.
\end{defi}
\begin{rem}
In contrast to the situation where $M$ is a globally hyperbolic spacetime 
\textit{with empty boundary} \cite{BGP}, existence, uniqueness and adjoint-relatedness 
of retarded/advanced Green's operators for the Klein-Gordon operator $P$ 
is in general not to be expected on spacetimes with timelike boundary. 
Positive results seem to be more likely on globally hyperbolic spacetimes 
with non-empty timelike boundary, although the general theory 
has not been developed yet to the best of our knowledge. 
Simple examples of adjoint-related pairs of Green's operators
were constructed e.g.\ in \cite{Casimir}.
\end{rem}

Given any region $V\in  \Reg_{M}[\Cauchy_M^{-1}]$ in $M$, 
which may intersect the boundary, 
we use the canonical inclusion $i : V\to M$ to define local retarded/advanced Green's operators
\begin{flalign}\label{eqn:GVdiagram}
\xymatrix{
C^\infty_\cc(\interior{V})\ar[d]_-{i_\ast} \ar[rr]^-{G^\pm_V} && C^\infty(\interior{V})\\
C^\infty_\cc(\interior{M}) \ar[rr]_-{G^\pm}&&  C^\infty(\interior{M}) \ar[u]_-{i^\ast}
}
\end{flalign} 
where $i_\ast$ denotes the pushforward of compactly supported functions (i.e.\ extension by zero) 
and $i^\ast$ the pullback of functions (i.e.\ restriction). Since $V\subseteq M$ is causally convex, 
it follows that $J^\pm_M(p) \cap V  = J^\pm_V(p)$ for all $p\in V$. 
Therefore $G^\pm_V$ satisfies the axioms of a retarded/advanced Green's operator for $P_V^{}$ on $V$. 
(Here we regard $V$ as a globally hyperbolic spacetime with timelike boundary, 
see Proposition \ref{propo:globhyp}. $J^\pm_V(p)$ 
denotes the causal future/past of $p$ in the spacetime $V$.) 
In particular, for all interior regions $V_{\mathrm{int}} \in \Reg_{\interior{M}}[\Cauchy_{\interior{M}}^{-1}]$ in $M$, 
by combining Proposition \ref{propo:globhyp} and \cite[Corollary 3.4.3]{BGP}
we obtain that $G^\pm_{V_{\mathrm{int}}}$ as defined in \eqref{eqn:GVdiagram} 
is the unique retarded/advanced Green's operator 
for the restricted Klein-Gordon operator $P_{V_{\mathrm{int}}}^{}$ 
on the globally hyperbolic spacetime $V_{\mathrm{int}}$ with empty boundary.
\sk

Consider any adjoint-related pair $(G^+,G^-)$ of Green's operator for $P$ on $M$.
For all $V\in\Reg_{M}[\Cauchy_M^{-1}]$, we set
$\III_{G^\pm} (V) \subseteq {\KKK}^{\ext}(V)$ to be the two-sided $\ast$-ideal
generated by the following relations:
\begin{itemize}
\item \textit{$G^\pm$-CCR:}  $\Phi_V(f) \, \Phi_V(g) - \Phi_V(g) \,\Phi_V(f) 
= \mathrm{i} \,\tau_V(f,g)~\oone$, for all $f,g\in C^\infty_\cc(\interior V)$, where
\begin{flalign}\label{eqn:tauV}
\tau_V(f,g) \,:=\, \int_V f \,G_V (g)~\vol_V
\end{flalign}
with $G_V := G_V^+ - G_V^-$ the causal propagator and $\vol_V$ the canonical volume form on $V$. 
\end{itemize}
The fact that the pair $(G^+,G^-)$ is adjoint-related (cf.\ Definition \ref{defi:Green}) implies
that for all $V\in\Reg_{M}[\Cauchy_M^{-1}]$ the causal propagator $G_V $ is formally skew-adjoint, hence $\tau_V$ is antisymmetric.
\begin{propo}
$\III_{G^\pm}\subseteq {\KKK}^{\ext}$ is an ideal that is trivial on the interior (cf.\ Definition \ref{defi:trivinteriorideal}).
\end{propo}
\begin{proof}
Functoriality of $\III_{G^\pm} : \Reg_{M}[\Cauchy_M^{-1}] \to\Vec$ is a consequence of
\eqref{eqn:GVdiagram}, hence $\III_{G^\pm}\subseteq {\KKK}^{\ext}$ is an ideal in the sense of
Definition \ref{def:ideal}. It is trivial on the interior because
for all interior regions $V_{\mathrm{int}}^{} \in \Reg_{\interior{M}}[\Cauchy_{\interior{M}}^{-1}]$,
the Green's operators defined by \eqref{eqn:GVdiagram} are the unique retarded/advanced Green's operators
for $P_{V_{\mathrm{int}}}^{}$ and hence the relations imposed by 
$\III_{G^\pm} (V_{\mathrm{int}}) $ automatically hold true in ${\KKK}^{\ext}(V_{\mathrm{int}})$ 
on account of the (partially-defined) CCR.
\end{proof}

\begin{rem}
We note that the results of this section still hold true
if we slightly weaken the hypotheses of Definition \ref{def:globhyp}
by assuming the strong causality and the compact double-cones 
property only for points in the interior $\interior M$ of $M$. 
In fact, $\interior M$ can still be covered 
by causally convex open subsets and 
any causally convex open subset $U \subseteq \interior M$ 
becomes a globally hyperbolic spacetime 
with empty boundary once equipped with the induced metric, 
orientation and time-orientation. 
\end{rem}

\begin{ex}
Consider the sub-spacetime $M := \bbR^{m-1} \times [0,\pi] \subseteq \bbR^m$ of $m$-dimensional
Minkowski spacetime, which has a timelike boundary $\partial M = \bbR^{m-1} \times \{0,\pi\}$.
The constructions in \cite{Casimir} define an adjoint-related pair $(G^+,G^-)$
of Green's operators for $P$ on $M$ that corresponds to Dirichlet boundary conditions.
Using this as an input for our construction above, we obtain a
quantum field theory ${\KKK}^{\ext}\big/\III_{G^\pm}\in\QFT(M)$ that may be interpreted
as the Klein-Gordon theory on $M$ with Dirichlet boundary conditions.
It is worth to emphasize that our theory in general \textit{does not} coincide with the one 
constructed in \cite{Casimir}. To provide a simple argument, let us focus
on the case of $m=2$ dimensions, i.e.\ $M = \bbR\times [0,\pi]$, and compare our 
global algebra $\AAA^{\mathrm{BDS}}(M) := {\KKK}^{\ext}\big/\III_{G^\pm} (M)$ with 
the global algebra $\AAA^{\mathrm{DNP}}(M)$ constructed in  \cite{Casimir}.
Both algebras are CCR-algebras, however the underlying symplectic 
vector spaces differ: The symplectic vector space underlying our global algebra 
$\AAA^{\mathrm{BDS}}(M)$ is $C^\infty_{\cc}(\interior{M})\big/ P C^\infty_\cc(\interior{M})$
with the symplectic structure \eqref{eqn:tauV}. Using that the spatial slices of $M = \bbR\times [0,\pi]$ 
are compact, we observe that the symplectic vector space underlying 
$\AAA^{\mathrm{DNP}}(M)$ is given by the space $\mathfrak{Sol}_{\mathrm{Dir}}(M)$ 
of {\em all} solutions with Dirichlet boundary condition on $M$ (equipped with the usual symplectic structure). 
The causal propagator defines a symplectic map $G : C^\infty_{\cc}(\interior{M})\big/ P C^\infty_\cc(\interior{M})
\to \mathfrak{Sol}_{\mathrm{Dir}}(M)$, which however is not surjective for the following reason:
Any $\varphi \in C^\infty_\cc(\interior{M})$ has by definition compact support in the interior of $M$,
hence the support of $G\varphi \in \mathfrak{Sol}_{\mathrm{Dir}}(M)$ is schematically as follows
\begin{flalign}
\begin{tikzpicture}
\draw[white,fill=gray!20] (-2,-2) -- (-2,2) -- (2,2) -- (2,-2) -- (-2,-2);
\node at (-2,-2.2) {{\small $x=0$}};
\node at (2,-2.2) {{\small $x=\pi$}};
\draw[black,fill=white] (-0.8,0) -- (-2,1.2) -- (-2,-1.2) -- (-0.8,0);
\draw[black,fill=white] (0.8,0) -- (2,1.2) -- (2,-1.2) -- (0.8,0);
\draw[line width=1mm] (-2,-2) -- (-2,2);
\draw[line width=1mm]  (2,-2) -- (2,2);
\draw[black,fill=gray!40] (0,0) ellipse (0.7 and 0.35);
\node at (0,0) {{\small $\mathrm{supp}\,\varphi$}};
\node at (0,1.5) {{\small $\mathrm{supp}\,G\varphi$}};
\end{tikzpicture}
\end{flalign}
The usual mode functions $\Phi_k(t,x) = \cos (\sqrt{k^2 + m^2}\,t)\, \sin (kx ) \in \mathfrak{Sol}_\mathrm{Dir}(M)$,
for $k\geq 1$, are clearly not of this form, hence $G : C^\infty_{\cc}(\interior{M})\big/ P C^\infty_\cc(\interior{M})
\to \mathfrak{Sol}_{\mathrm{Dir}}(M)$ cannot be surjective. As a consequence,
the models constructed in \cite{Casimir} are in general not additive from the interior
and our construction ${\KKK}^{\ext}\big/\III_{G^\pm}$ should be interpreted as
the maximal additive subtheory of these examples. It is interesting to note that 
there exists a case where both constructions coincide: Consider the sub-spacetime 
$M:=\mathbb{R}^{m-1}\times [0,\infty)\subseteq \bbR^m$ 
of Minkowski spacetime with $m \geq 4$ even and take a
massless real scalar field with Dirichlet boundary conditions. 
Using Huygens' principle and the support properties of the 
Green's operators one may show that our algebra $\AAA^{\mathrm{BDS}}(M)$ 
is isomorphic to the construction in \cite{Casimir}.
\end{ex}


\section*{Acknowledgments}
We would like to thank Jorma Louko for useful discussions
about boundaries in quantum field theory and also Didier~A.~Solis 
for sending us a copy of his PhD thesis \cite{Solis}.
The work of M.B.\ is supported by a research grant funded by 
the Deutsche Forschungsgemeinschaft (DFG, Germany). 
Furthermore, M.B.\ would like to thank 
Istituto Nazionale di Alta Matematica (INdAM, Italy) 
for the financial support provided during his visit 
to the School of Mathematical Sciences 
of the University of Nottingham, 
whose kind hospitality is gratefully acknowledged.
The work of C.D.\ is supported by
the University of Pavia. C.D.\ gratefully acknowledges the kind of hospitality of the 
School of Mathematical Sciences of the University of Nottingham, where part of 
this work has been done.
A.S.\ gratefully acknowledges the financial support of 
the Royal Society (UK) through a Royal Society University 
Research Fellowship, a Research Grant and an Enhancement Award. 
He also would like to thank the Research Training Group 1670 ``Mathematics 
inspired by String Theory and Quantum Field Theory'' for the invitation 
to Hamburg, where part of this work has been done.


\appendix 

\section{\label{app:cattheory}Some concepts from category theory}
\subsection*{Adjunctions:} This is a standard concept, 
which is treated in any category theory textbook, e.g.\ \cite{Riehl}.
\begin{defi}\label{def:adjunction}
An  \textit{adjunction} consists of a pair of functors
\begin{flalign}
\xymatrix{
F : \CC \ar@<0.5ex>[r] & \DD : G \ar@<0.5ex>[l] \quad,
}
\end{flalign}
together with natural transformations $\eta: \id_\CC \to GF$ (called \textit{unit})
and $\epsilon : FG\to \id_\DD$ (called \textit{counit}) that satisfy 
the \textit{triangle identities}
\begin{flalign}
\xymatrix{
\ar[dr]_-{\id_F} F \ar[r]^-{F\eta} & FGF\ar[d]^-{\epsilon F}\\
& F
}\qquad\qquad
\xymatrix{
\ar[dr]_-{\id_G}G\ar[r]^-{\eta G } & GFG\ar[d]^-{G \epsilon}\\
&G
}
\end{flalign}
We call $F$ the \textit{left adjoint} of $G$ and 
$G$ the \textit{right adjoint} of $F$, and write $F\dashv G$.
\end{defi}

\begin{defi}\label{def:adjointequivalence}
An  \textit{adjoint equivalence} is an adjunction
\begin{flalign}
\xymatrix{
F : \CC \ar@<0.75ex>[r]_-{\sim} & \DD : G \ar@<0.75ex>[l] 
}
\end{flalign}
for which both the unit $\eta$ and the counit $\epsilon$ are natural isomorphisms.
Existence of an adjoint equivalence in particular implies 
that $\CC\cong\DD$ are equivalent as categories.
\end{defi}

\begin{propo}\label{propo:adjointunique}
If a functor $G :\DD\to \CC$ admits a left adjoint $F :\CC\to\DD$, then $F$ is unique up to
a unique natural isomorphism. Vice versa, if a functor $F : \CC\to \DD$ admits
a right adjoint $G : \DD\to \CC$, then $G$ is unique up to a unique natural isomorphism.
\end{propo}

\subsection*{Localizations:}
Localizations of categories are treated for example in \cite[Section 7.1]{KashiwaraSchapira}.
In our paper we restrict ourselves to small categories.
\begin{defi}\label{def:localization}
Let $\CC$ be a small category and $W\subseteq \mathrm{Mor}\,\CC$ a subset of the set of morphisms.
A \textit{localization of $\CC$ at $W$} is a small category $\CC[W^{-1}]$ together with 
a functor $L :\CC\to \CC[W^{-1}]$ satisfying the following properties:
\begin{itemize}
\item[(a)] For all $(f :c \to c^\prime)\in W$, $L(f) : L(c)\to L(c^\prime)$ is an isomorphism in $\CC[W^{-1}]$.

\item[(b)] For any category $\DD$ and any functor $F: \CC\to \DD$ that sends morphisms in $W$
to isomorphisms in $\DD$, there exists a functor $F_W : \CC[W^{-1}]\to \DD$ and a natural isomorphism
$F \,\cong\, F_W \, L$.

\item[(c)] For all objects $G,H\in \DD^{\CC[W^{-1}]}$ in the functor category, the map
\begin{flalign}
\Hom_{\DD^{\CC[W^{-1}]}}\big(G,H\big) \longrightarrow \Hom_{\DD^{\CC}}\big(G\, L ,H\, L\big)
\end{flalign}
is a bijection of $\Hom$-sets.
\end{itemize}
\end{defi}
\begin{propo}
If $\CC[W^{-1}]$ exists, it is unique up to equivalence of categories.
\end{propo}


\end{document}